\documentclass[11pt,letterpaper]{article}
\usepackage{graphicx} 
\usepackage{color}
\usepackage[margin=1in]{geometry}
\usepackage{fluff}
\usepackage{mathtools}
\usepackage{xcolor}
\usepackage[suppress]{color-edits}
\usepackage{hyperref}
\addauthor{er}{orange}
\addauthor{ex}{purple}
\addauthor{smw}{blue}
\addauthor{md}{brown}

\definecolor{mygreen}{RGB}{10,110,230}
\definecolor{myred}{RGB}{10,110,230}

\hypersetup{
     colorlinks=true,
     citecolor= mygreen,
     linkcolor= red
}

\newcommand{\srev}{\textsc{SRev}}
\newcommand{\vcg}{\textsc{VCG}}
\newcommand{\rev}{\textsc{Rev}}

\newcommand{\ERT}[1][T]{\ensuremath{\mathcal{ER}_{\leq#1}}}
\newcommand{\comp}{\textsc{Comp}}
\newcommand{\sscomp}{\textsc{SSComp}}
\newcommand{\mattnote}[1]{\textcolor{blue}{#1}}

\newcommand{\ER}{\ensuremath{\mathcal{ER}}}

\newcommand{\hide}[1]{}

\title{Settling the Competition Complexity of Additive Buyers over Independent Items}
\author{
    Mahsa Derakhshan\thanks{Northeastern University, \texttt{m.derakhshan@northeastern.edu}}
    \and
    Emily Ryu\thanks{Cornell University, \texttt{eryu@cs.cornell.edu}}
    \and
    S. Matthew Weinberg\thanks{Princeton University, \texttt{smweinberg@princeton.edu}}
    \and
    Eric Xue\thanks{Princeton University, \texttt{ex3782@princeton.edu}}
}

\begin{document}

\maketitle

\begin{abstract}
    The competition complexity of an auction setting is the number of additional bidders needed such that the simple mechanism of selling items separately (with additional bidders) achieves greater revenue than the optimal but complex (randomized, prior-dependent, Bayesian-truthful) optimal mechanism without the additional bidders. Our main result settles the competition complexity of $n$ bidders with additive values over $m<n$ independent items at $\Theta(\sqrt{nm})$. The $O(\sqrt{nm})$ upper bound is due to \cite{BeyhaghiW19}, and our main result improves the prior lower bound of $\Omega(\ln n)$ to $\Omega(\sqrt{nm})$. 

Our main result follows from an explicit construction of a Bayesian IC auction for $n$ bidders with additive values over $m<n$ independent items drawn from the Equal Revenue curve truncated at $\sqrt{nm}$ ($\ERT[\sqrt{nm}]$), which achieves revenue that exceeds $\srev_{n+\sqrt{nm}}(\ERT[\sqrt{nm}]^m)$. 

Along the way, we show that the competition complexity of $n$ bidders with additive values over $m$ independent items is \emph{exactly equal to} the minimum $c$ such that $\srev_{n+c}(\ERT[p]^m) \geq \rev_n(\ERT[p]^m)$ for all $p$ (that is, \emph{some} truncated Equal Revenue witnesses the worst-case competition complexity). Interestingly, we also show that the untruncated Equal Revenue curve does \emph{not} witness the worst-case competition complexity when $n > m$: $\srev_n(\mathcal{ER}^m) = nm+O_m(\ln (n)) \leq \srev_{n+O_m(\ln (n))}(\mathcal{ER}^m)$, and therefore our result can only follow by considering all possible truncations.
\end{abstract}

\addtocounter{page}{-1}

\newpage

\section{Introduction}

Multi-dimensional mechanism design has become a core subdomain of TCS following the seminal work of Chawla, Hartline, and Kleinberg, introducing its study to Computer Science~\cite{ChawlaHK07}. In particular, while Myerson's seminal work in the \emph{single}-dimensional setting elegantly characterizes the optimal single-item auction in quite broad settings, Economists and Computer Scientists alike soon realized that optimal mechanisms in the multi-dimensional setting, even in restricted two-item instances, can be horribly intractable~\cite{RochetC98, BriestCKW15, HartN13, HartR15, Thanassoulis04, Pavlov11, DaskalakisDT17, PsomasSW19, PsomasSW22, WeinbergZ22}. In response,~\cite{ChawlaHK07} initiates a vast series of works establishing that simple mechanisms, while rarely optimal, achieve constant-factor approximations in quite rich settings~\cite{ChawlaHK07, ChawlaHMS10, ChawlaMS15, HartN17, LiY13, BabaioffILW20, Yao15, RubinsteinW15, ChawlaM16, CaiZ17, EdenFFTW21}. These works help explain the prevalence of simple auctions in practice.

Still, constant-factor approximations do not tell the whole story -- sticking with something simple that guarantees \emph{something} is a reasonable starting point, but why not shoot for more? The Resource Augmentation paradigm offers a different perspective: running a complex auction is costly -- is it perhaps more cost-effective instead to recruit extra bidders (the ``resources'') to participate in a simple auction? That is, prior-dependent (versus prior-independent) mechanisms are costly because you must learn the prior. Bayesian IC, BIC (versus Dominant Strategy IC, DSIC) mechanisms are costly because you must, at minimum, teach bidders the concept of Bayes-Nash equilibria (or set up auto-bidding infrastructure and convince them to trust it, etc.). Randomized mechanisms are costly because you must further ensure the risk-neutrality of your bidders. Computationally intractable mechanisms are costly simply because computation is expensive. What if recruiting extra bidders for a  prior-independent, DSIC, deterministic, computationally tractable mechanism could outperform the complex optimum (without additional bidders) -- might that be more cost-effective?


The mathematical question at hand, then, is to nail down \emph{how many additional bidders are necessary for a simple auction to outperform the (intractable) optimum?} The seminal work of Bulow and Klemperer~\cite{BulowK96} is the first to ask such a question and establish that the (prior-independent, DSIC, deterministic, computationally efficient) second-price auction with one additional bidder outperforms Myerson's (prior-dependent, DSIC, deterministic, computationally-efficient) revenue-optimal auction in single-item settings with regular\footnote{A single-variate distribution $F$ is regular if the function $\varphi_F(x):=x - \frac{1-F(x)}{f(x)}$ is monotone non-decreasing.} bidders.\footnote{In this setting, Myerson's optimal auction is exceptionally simple: it is just a second-price auction with reserve. So~\cite{BulowK96} essentially argues that one additional bidder removes the need for prior dependence and does not provide commentary on BIC vs.~DSIC, randomized vs.~deterministic, or computational tractability.}~\cite{RoughgardenTY12} are the first to ask this question in multi-dimensional settings, and~\cite{EdenFFTW17b} term the minimum number of bidders needed the \emph{competition complexity}. Specifically, for a class $\mathcal{C}$ of distributions over valuation functions for a single bidder, the competition complexity $\comp_{\mathcal{C}}(n):=\inf_{c \in \mathbb{N}_{\geq 0}}\{ c\ |\ \vcg_{n+c}(D) \geq \rev_n(D)\ \forall\ D \in \mathcal{C}\}$.\footnote{$\mathcal{C}$ is a class of distributions such as ``additive over $m$ independent items.'' $D$ is a distribution such as ``the value for item $j$ is drawn independently from $U([0,j])$.'' $\rev_n(D)$ denotes the optimal revenue of any BIC auction for $n$ bidders drawn iid from $D$, and $\vcg_n(D)$ denotes the expected revenue of the welfare-maximizing VCG mechanism~\cite{Vickrey61, Clarke71, Groves73} for $n$ bidders drawn iid from $D$. See Section~\ref{sec:prelim} for further clarity.} 

In the canonical domain of $n$ additive bidders over $m$ independent items, (the same domain studied in~\cite{HartN17, LiY13, BabaioffILW20, Yao15, CaiDW16, EdenFFTW17b, FeldmanFR18, BeyhaghiW19}),~\cite{EdenFFTW17b} first establish a competition complexity bound of at most $n+2(m-1)$. That is, if $\mathcal{A}_m^R$ denotes the class of all valuation distributions which are additive across items, and each item valuation is drawn independently from a regular distribution, then $\comp_{\mathcal{A}_m^R}(n) \leq n+2(m-1)$. In other words, the VCG mechanism with $n+2(m-1)$ additional bidders outperforms the optimum (without additional bidders) for any distribution $D \in \mathcal{A}_m^R$. 

In the ``little $n$ regime'' ($n = O(m)$), their bound was later improved to $\comp_{\mathcal{A}_m^R}(n) = \Theta(n\ln(2+m/n))$, which is tight (up to constant factors)~\cite{FeldmanFR18, BeyhaghiW19}. In the ``big $n$ regime'' ($n = \Omega(m)$),~\cite{BeyhaghiW19} establish that $\comp_{\mathcal{A}_m^R}(n) \in [\Omega(\ln n), 9\sqrt{nm}]$, leaving open an exponential gap.
Our main result closes the final gap in the ``Big $n$'' regime: the competition complexity is indeed $\Theta(\sqrt{nm})$. That is, when $m\ge 2$:\\

\noindent\textbf{Main Result:} $\comp_{\mathcal{A}_m^R}(n) = \Omega(\sqrt{nm})$: The competition complexity of $n$ additive bidders over $m$ independent items is $\Omega(\sqrt{nm})$ in the ``Big $n$'' regime. Combined with~\cite{BeyhaghiW19}, this settles $\comp_{\mathcal{A}_m^R}(n) = \Theta(\sqrt{nm})$ in this regime (the ``little $n$'' regime is previously settled~\cite{FeldmanFR18, BeyhaghiW19}). \\

While our main result ultimately follows by designing a single BIC auction with high revenue, we highlight areas of technical interest briefly below. Similarly, while much of the journey towards our main result is not ``necessary'' for its final proof, several aspects of the journey are likely of independent interest, and we highlight these subsequently.

\subsection{Main Result: Technical Highlights}\label{sec:mainoverview}
Our main result ultimately follows by designing a BIC auction for $n$ bidders whose values for $m$ items are drawn from the Equal Revenue curve truncated at $T = \Theta(\sqrt{nm})$ ($\ERT$).\footnote{$\ERT$ has CDF $F(x) = 1-1/x$ for $x \in [1, T)$, $F(T) = 1$, and $F(x) = 0$ for $x < 1$.} A priori, it is unclear what might be technically engaging about designing a BIC auction for a particular distribution. We briefly overview three technical highlights:
\begin{itemize}
\item When $n \gg m$ (the regime we study), selling each of $m$ items separately to $n$ bidders whose value for each item is drawn from $\ERT$ \emph{already achieves expected revenue extremely close to the expected welfare}. To see this, observe that the expected welfare is clearly at most $mT$, and selling items separately using a second-price auction with reserve $T$ achieves revenue $\geq mT\cdot (1-e^{-n/T}) = mT (1-e^{-\Omega(\sqrt{n/m})})$. This means there is little room for a more sophisticated auction to outperform selling separately without additional bidders, let alone with additional bidders.\footnote{But, it does also mean that we may not need to outperform selling separately by much in order to also outperform selling separately with additional bidders, because the additional bidders cannot possibly help much either.}
\item The auction we design is not DSIC -- we explicitly design an interim allocation rule together with interim payments and prove that the mechanism is implementable and BIC. To the best of our knowledge, there are not many prior instances of useful explicit designs of BIC-but-not-DSIC mechanisms -- the only notable example is that of~\cite{Yao17} for two bidders whose values for each of two items are drawn independently from the uniform distribution on $\{1,2\}$. We design such an auction for any $m\geq 2$ and $n \geq m$.
\item The method we use to design our auction is likely of use to future designs of BIC-but-not-DSIC auctions. We start by picking an allocation rule we would like to implement and prices we would like to charge that result in a clean analysis. Unfortunately, the prices we'd like to charge are not BIC, but the interim allocation rule (of our desired allocation rule) is well-structured, so the prices can be massaged to get fairly close to BIC. Further unfortunately, there does not appear to be an exactly-BIC implementation of this allocation rule at all, so we further slightly massage the allocation rule. This last step, in particular, is reminiscent of a specialized (for our mechanism and distribution) instantiation of an $\varepsilon$-BIC to BIC reduction~\cite{HartlineKM11, BeiH11, DaskalakisW12, RubinsteinW15, DughmiHKN17, CaiOVZ21}.
\end{itemize}

\subsection{Results of Independent Interest along the Journey}\label{sec:journey}
The journey towards our main result yields two additional results -- while these are ultimately ``unnecessary'' steps along the journey, they help provide context for our results and approach.\\

\noindent\textbf{Independent Result I:} Let $\mathcal{ERC}_m$ denote the subclass of $\mathcal{A}_m^R$ containing only distributions of the form $\ERT^m$ for some truncation $T$. Then $\comp_{\mathcal{A}^R_m}(n) = \comp_{\mathcal{ERC}_m}(n)$ for all $n$. That is, for all $n$, the worst-case competition complexity of any distribution in $\mathcal{A}^R_m$ is witnessed by iid truncated Equal Revenue curves.\\

One direction of this equality is trivial, as $\mathcal{ERC}_m \subset \mathcal{A}_m^R$. The non-trivial direction does not at all follow from identifying an a priori worst-case distribution (indeed, the worst-case distribution for $\comp_{\mathcal{A}_m^R}(n)$ is $\mathcal{ER}_{\leq \Theta(\sqrt{mn})}$, which has no apparent a priori justification). It is natural to guess that an untruncated equal revenue curve may be the worst-case distribution, as it stochastically dominates all other distributions with the same single-bidder revenue. However, this intuition breaks rather quickly: (a) just because every marginal of $D$ stochastically dominates those of $D'$ does \emph{not} imply that the optimal revenue for $D$ exceeds that of $D'$ due to the phenomenon of revenue non-monotonicity~\cite{HartR15, RubinsteinW15, Yao18}, (b) even if moving from $D'$ to $D$ were guaranteed to make the revenue benchmark larger, it also improves the revenue of selling separately, so both sides of the desired inequality increase. Instead, we show that a modification of ~\cite{BeyhaghiW19}'s approach upper bounds $\comp_{\mathcal{A}_m^R}(n)$ \emph{if and only if it upper bounds $\comp_{\mathcal{ERC}_m}(n)$}. See Section~\ref{sec:ERTreduction} for further details.

In addition, Independent Result I provides further context for our main result. A practically-minded reader might wonder why it matters that  $\Omega(\sqrt{mn})$ bidders are necessary for VCG to outperform the optimum in an instance where selling separately is already extremely close to optimal.\footnote{See Section~\ref{sec:related} for a very brief note on results such as~\cite{FeldmanFR18, CaiS21} that explicitly consider resource augmentation to target a $(1-\varepsilon)$-approximation rather than truly exceeding the optimum.} Independent Result I highlights that analyzing this instance is in some sense a necessary step to analyze instances where the gap might be larger.\\ 

\noindent\textbf{Independent Result II:} For some absolute constant $c$, $\vcg_{n+cm\ln(n)}(\ER^m) \geq \rev_n(\ER^m)$. That is, $O(m\ln(n))$ additional bidders suffice for selling separately $n$ bidders with values for $m$ items drawn from iid untruncated Equal Revenue curves to exceed the optimal revenue. This implies that the untruncated Equal Revenue curve is \emph{not} the worst-case instance for \emph{any} $n \geq m$ (and it witnesses a lower bound that is exponentially suboptimal in $n$). This result follows by establishing that the optimal revenue for $n$ bidders with values for each of $m$ items drawn iid from the Equal Revenue curve is $nm+O(m^2\ln(n))$. \\


Previous disclaimers that an untruncated revenue curve is not \emph{obviously} the worst-case distribution notwithstanding, it is still a tempting conjecture that the equal revenue curve may indeed be the worst-case (or at least, asymptotically close) -- if one had hoped to improve~\cite{BeyhaghiW19}'s bounds, iid untruncated equal revenue curves is a natural first step. So it is interesting that the untruncated Equal Revenue curve witnesses an exponentially-suboptimal bound compared to a properly-truncated Equal Revenue curve. This provides further motivation for Independent Result I, as it highlights that there is indeed no a priori worst-case distribution. See Section~\ref{sec:upperbound} for further details.


Beyond the result itself, our analysis should be of independent interest. In particular, a first step towards our upper bound on the revenue is a flow in the~\cite{CaiDW16} framework. To the best of our knowledge, prior works that approximate the optimal revenue all use a ``region-separated'' flow (see Section~\ref{sec:kfbic} for a formal definition)~\cite{CaiDW16, CaiZ17, EdenFFTW17b, EdenFFTW21, BeyhaghiW19, CaiOZ22}.\footnote{There are certainly works, such as~\cite{HartlineH15, DaskalakisDT17} that use more complex flows to derive \emph{optimal} mechanisms in single-bidder settings.} We prove that such a flow cannot possibly witness an upper bound better than $nm+\Omega(m\sqrt{nm})$ on $\rev_n(\ER^m)$ by designing an auction that satisfies all ``within-region'' BIC constraints (but not the cross-region constraints). To cope with this, our analysis still begins from a region-separated flow (in fact, the same canonical flow used in~\cite{EdenFFTW17b}), but adds a novel second step to (necessarily) leverage cross-region constraints.

Finally, in Section~\ref{sec:ER-revenue-LB} we also establish that bundling items together (to $n$ bidders with $m$ items from the Equal Revenue curve) with a second-price auction achieves expected revenue $nm+\Omega(m\ln(n))$. Our analysis follows primarily from a coupling argument as opposed to raw calculations, and also slightly improves the analysis of~\cite{BeyhaghiW19} from $nm+\Omega(\ln(n))$. 


\subsection{Related Work}\label{sec:related}


We have already overviewed the most directly related work.~\cite{BulowK96} first consider resource augmentation for Bayesian mechanism design, and establish that a single additional bidder suffices for the second-price auction with no reserve to outperform the revenue-optimal auction with any number of i.i.d.~regular bidders and a single item.~\cite{RoughgardenTY12} first consider resource augmentation for multi-dimensional mechanism design, and compare the VCG mechanism with additional bidders to the optimal deterministic DSIC auction for unit-demand bidders over independent items.~\cite{EdenFFTW17b} are the first to target outperformance of the optimal BIC randomized auction, and study the now-canonical setting of additive bidders over independent items. Their bounds have since been tightened by~\cite{FeldmanFR18, BeyhaghiW19}, and our main result tightens the last remaining gap. Moreover, if $\mathcal{A}^R_{m,\mathcal{I}}$ denotes the class of valuation functions that are additive over $m$ independent regular items subject to downwards-closed constraints $\mathcal{I}$,~\cite{EdenFFTW17b} also establish that $\comp_{\mathcal{A}^R_{m,\mathcal{I}}}(n) \leq \comp_{\mathcal{A}_m^R}(n) + m-1$. Therefore, the study of additive buyers has implications for significantly more general settings as well. Other works, such as~\cite{LiuP18, FuLR19, BrustleCDV22} consider the competition complexity of Bayesian mechanism design in other settings (such as dynamic auctions, non-iid single-dimensional bidders, or posted-price mechanisms).~\cite{FeldmanFR18, CaiS21} further consider how many additional bidders are needed to recover a $(1-\varepsilon)$-fraction of the optimal revenue, rather than truly exceeding the optimal revenue. Both results require strictly fewer bidders than would otherwise be necessary.

The concept of resource augmentation is well-represented within TCS broadly~\cite{SleatorT85, BarmanCU12}, Economics broadly~\cite{AkbarpourMS18, AkbarpourALS22}, and also their intersection~\cite{RoughgardenT02, ChawlaHMS13}.

We have also previously noted a vast literature justifying simple auctions in multi-dimensional settings, despite their suboptimality~\cite{ChawlaHK07, ChawlaHMS10, ChawlaMS15, HartN17, LiY13, BabaioffILW20, Yao15, RubinsteinW15, ChawlaM16, CaiZ17, EdenFFTW21}.\footnote{Note that~\cite{HartlineR09} initiate a conceptually-similar line of work justifying exceptionally simple auctions in single-dimensional settings via constant-factor approximation guarantees.} Our independent results use similar technical tools (such as the benchmark induced by~\cite{CaiDW16}'s ``canonical flow''), but deviates from these in seeking a $(1-o(1))$-approximation to the optimal revenue, rather than a constant-factor approximation. 

At a technical level, our results are similar-to-yet-distinct-from several themes in the literature on multi-dimensional mechanism design. As previously noted, we design an explicit BIC auction that is not DSIC, which is also done in~\cite{Yao17}. A difference is that~\cite{Yao17} considers two bidders and two items and proves that the BIC auction strictly outperforms the optimal DSIC auction, whereas we consider arbitrarily-many bidders and items but do not explicitly compare to a DSIC auction. We have also mentioned that one step of our auction design bears similarity to $\varepsilon$-BIC to BIC reductions, which are developed in~\cite{DaskalakisW12, RubinsteinW15, CaiOVZ21} based on techniques introduced in~\cite{HartlineKM11, BeiH11, DughmiHKN17}. Their results apply generally and are technically quite involved, whereas we directly massage a specific nearly-BIC auction for a specific distribution. There is also a line of works deriving \emph{optimal} mechanisms for specific distributions~\cite{DaskalakisDT17, GiannakopoulosK14, GiannakopoulosK15, HartlineH15}. These works consider single bidder settings, and most use some form of duality to establish optimality. Other works derive optimal mechanisms for simple classes of single-bidder distributions to establish computational hardness~\cite{DaskalakisDT14, ChenDOPSY22}. In comparison, our work considers multi-bidder settings, and in some sense lies between these works and constant-factor approximations in terms of complexity: our upper bounds are slightly more involved than those sufficient for constant-factor approximations, but not as involved as those necessary for precise optimality. At the same time, we do not nail precisely the revenue-optimal auctions, but do derive bounds strictly better than what can be achieved by the simple duals sufficient for constant-factors.

\section{Preliminaries} \label{sec:prelim}
In this section, we provide the minimal preliminaries necessary to state and prove our main result. Section~\ref{sec:addprelims} provides additional preliminaries specific to our independent results. 

The setting we study consists of $n$ bidders with additive valuations over $m$ items.
Formally, the values of the bidders are drawn from an $n \times m$ dimensional joint distribution $\mathcal{D}$ where $v_{ij}$ denotes bidder $i$'s value for item $j$.
Bidder $i$'s value for a subset $S$ of items is $\sum_{j \in S} v_{ij}$.

A mechanism is given by ex-post allocation and payment rules that specify the probabilities with which each bidder gets each item and how much each bidder pays for each item, respectively. We will also consider interim allocation rules of auctions, which suffice to understand whether mechanisms are Bayesian IC (see Definition~\ref{def:BIC} below).

\begin{definition}
The ex-post probabilities with which each bidder receives each item are given by a function $x: \supp{\mathcal{D}} \to \Delta^{n \times m}$ where $x_{ij}(v_1, \dots, v_n)$ denotes the probability with which bidder $i$ receives item $j$ given the bid profile $v$, and the ex-post payments $q: \supp{\mathcal{D}} \to \Delta^{n \times m}$ has $q_i(v)$ denote the payment that bidder $i$ makes given the bid profile $v$. Given an ex-post allocation rule $x$ and price rule $q$, the interim probability with which bidder $i$ receives item $j$ when she bids $v_i$ and the interim price paid are defined as
\[
    \pi_{ij}(v_i) := \underset{v \sim \mathcal{D}}{\EE}\cSqBr{x_{ij}(v)}{v_i} \qquad \qquad p_i(v_i):=\underset{v \sim \mathcal{D}}{\EE}\cSqBr{q_i(v)}{v_i} 
\]
That is, the interim probability is the expected probability with which bidder $i$ receives item $j$ when she bids $v_i$ and the remaining bidders bid truthfully, and the interim price is the expected price when bidding $v_i$ and the remaining bidders bid truthfully.
\end{definition}

\begin{definition}\label{def:BIC}
Let $\pi$ denote an interim allocation rule and let $p$ denote an interim payment rule.
The mechanism $(\pi, p)$ is Bayesian Incentive Compatible (BIC) if for all $i, v_i, v_i'$,
\[
    \sum_j \pi_{ij}(v_i) v_{ij} - p_i(v_i) \geq \sum_j \pi_{ij}(v_i') v_{ij} - p_i(v_i')
\]
That is, each bidder's best response to her peers if they report their true values is also to report her true values.
\end{definition}

In addition, we use the following terminology:

\begin{itemize}
    \item $\ER$: the single-variate distribution with CDF $F(x) = 1 - \frac{1}{x}$, for $x \ge 1$.
    \item $\ER^m$: the multi-variate distribution that draws $m$ values i.i.d.~from $\ER$. 
    \item $\ER^{n\times m}$: the multi-variate distribution drawing $n$ bidders' values for $m$ items i.i.d.~from $\ER$. 
    \item $\ERT$: the equal revenue distribution truncated at $T$; i.e. the single-variate distribution with CDF $F(x) = 1 - \frac{1}{x}$ for $x \in [1,T)$ and $F(T) = 1$.
    \item $\rev^{M}(\mathcal{D})$: the expected revenue of an auction $M$ when played by bidders drawn from the joint distribution $\mathcal{D}$ over values of $n$ bidders for $m$ items.
    \item $\rev(\mathcal{D})$: the supremum over all BIC auctions $M$ of $\rev^M(\mathcal{D})$.
    \item $\srev(\mathcal{D})$: the expected revenue of selling separately (using Myerson's optimal auction~\cite{Myerson81}) to bidders drawn from the joint distribution $\mathcal{D}$ over values of $n$ bidders for $m$ items.
    \item $\vcg(\mathcal{D})$: the expected revenue of the welfare-maximizing Vickrey-Clarke-Groves (VCG) auction when played by bidders drawn from the joint distribution $\mathcal{D}$ over values of $n$ bidders for $m$ items.
    \item All $\mathcal{D}$ considered in this paper are i.i.d.~across bidders, and of the form $D^n$ for some distribution over a single bidder's valuation function for $m$ items. To simplify notation throughout, in these cases we notate $\rev_n(D):=\rev(\mathcal{D}), \srev_n(D):=\srev(\mathcal{D}), \vcg_n(D):=\vcg(\mathcal{D})$.
\end{itemize}

Below is a formal (re-)statement of the competition complexity. Our work considers additive bidders, where the VCG auction sells items separately using a second-price auction (so the distinction between VCG and selling separately is simply whether or not there is a reserve, when item values are regular, or whether items are sold using a second-price auction vs.~Myerson's optimal auction in the general case). 

\begin{definition}[Competition Complexity] Let $\mathcal{C}$ be a class of distributions over valuation functions for a single bidder. The \emph{Competition Complexity} of $\mathcal{C}$ is the function $\comp_{\mathcal{C}}(\cdot):\mathbb{N}_+ \rightarrow \mathbb{N}_{\geq 0}$ where $\comp_{\mathcal{C}}(n):=\inf_{c \in \mathbb{N}_{\geq 0}}\{c\ |\ \vcg_{n+c}(D) \geq \rev_n(D)\ \forall\ D \in \mathcal{C}\}$. The \emph{Selling Separately Competition Complexity} is instead $\sscomp_{\mathcal{C}}(n):=\inf_{c \in \mathbb{N}_{\geq 0}}\{c\ |\ \srev_{n+c}(D) \geq \rev_n(D)\ \forall\ D \in \mathcal{C}\}$.
\end{definition}


\section{Main Result: $\comp_{\mathcal{A}^R_m}(n)=\Omega(\sqrt{nm})$} \label{sec:ERTsqrt}

For the class of truncated equal revenue distributions $\ERT$ with $T = \lambda \sqrt{nm}$ for some absolute constant $\lambda > 1$ (and $T < n$), we now provide an explicit construction of a BIC auction $M$ with $\rev^M(\ERT^{n\times m}) > \srev_{n+c\sqrt{nm}}(\ERT^m)$ for some absolute constant $c$. This witnesses that $\comp_{\mathcal{A}^R_m}(n) = \Omega(\sqrt{nm})$, and $\sscomp_{\mathcal{A}_m}(n) = \Omega(\sqrt{nm})$.

For ease of readability, several calculation-based proofs are deferred to Appendix~\ref{sec:proofs_ERTsqrt}.

\subsection{Step One: Intuition \& a Not-at-all BIC Auction}
First, we explicitly compute $\srev_{n'}(\ER^m)$.

\begin{restatable}{lemma}{srevcompute} 
$\srev_{n'}(\ERT^m) = m\cdot T \cdot \left(1-\left(1-1/T\right)^{n'}\right)$. One mechanism achieving this sells each item separately with a second-price auction at reserve $T$.\footnote{In fact, a second-price auction with any reserve $\leq T$ achieves this.}
\end{restatable}

Note that selling separately is already nearly optimal when $T \ll n$. In particular, $\rev_n(\ERT^m) \leq mT$, and $\srev_{n}(\ERT^m) \approx mT$ when $T \ll n$. Moreover, for each item $j$, selling separately achieves the maximum value whenever it is $T$, and so the only possible room for improvement over selling separately is in the exponentially-unlikely cases that all bidders have value $<T$ for item $j$ (where selling separately gets $0$ revenue from item $j$, yet there is strictly positive value). 

So, in order to possibly have an auction whose revenue exceeds $\srev_n(\ERT^m)$ (let alone $\srev_{n+c\sqrt{nm}}(\ERT^m)$), we must somehow get non-zero revenue from item $j$ in cases when $v_{ij} < T$ for all $i$. One very naive way to accomplish this is simply to remove the reserve, and sell each item with a second-price auction instead. Of course, this is still selling separately, and therefore offers no improvement over $\srev_n(\ERT^m)$. But, it highlights the tradeoff that any BIC mechanism must face: allocating item $j$ to a bidder $i$ with value $v_{ij} < T$ provides incentive for bidder $i$ to misreport that $v_{ij} < T$ when in fact $v_{ij} = T$, and therefore risks revenue $<T$ in cases where selling separately achieves $T$. 

So, the first idea in designing our BIC auction is to find opportunities to allocate item $j$ to a bidder $i$ with $v_{ij}<T$ without risking too much in cases where $v_{ij} =T$ instead. Below is an allocation rule that accomplishes this first step, but is not yet BIC. In particular, we only ever consider allocating an item $j$ to a bidder $i$ with $v_{ij}<T$ \emph{if $v_{ij'} = T$ for at least one other item $j'$}.

\begin{definition}[The Naive Auction]
The Naive Auction allocates each item $j$ separately as follows.
\begin{enumerate}
    \item If there exists a bidder $i$ with $v_{ij} = T$, then allocate item $j$ uniformly at random to such a bidder and charge a price of $T$.
    \item If $v_{ij} < T$ for all $i$, but there exists a bidder $i$ with both $v_{ij'} = T$ for some $j' \neq j$ and $v_{ij} \geq mn/T$, then allocate item $j$ uniformly at random to such a bidder and charge a price of $mn/T$.
    \item Otherwise, do not allocate or elicit payments for item $j$.
\end{enumerate}
\end{definition}

The Naive Auction is certainly not BIC: a bidder whose values are $T$ for every item achieves utility of $0$ for reporting the truth, but $>0$ for instead lowering one value to $mn/T$. Still, it clearly achieves revenue greater than $\srev_n(\ERT^m)$. We first establish that the revenue of the Naive Auction further exceeds $\srev_{n+c\sqrt{nm}}(\ERT^m)$ -- the remainder of this section is then devoted to massaging the Naive Auction into a BIC auction without losing much of this additional revenue.

\begin{lemma}\label{lem:srev}
    $\srev_{n+x}(\ERT^m) \leq \srev_n(\ERT^m) + mx(1-1/T)^n$.
\end{lemma}
\begin{proof}
    Couple draws from $\ERT^{m\times (n+x)}$ and $\ERT^{m\times n}$ so that the first $n$ bidders' values are identical. For each item $j$, selling separately to $n+x$ bidders outperforms selling separately to $n$ bidders iff $v_{ij} < T$ for all $i \in [n]$ and $v_{ij} = T$ for some $i > n$ (and when this occurs, it outperforms by exactly $T$).

    Therefore, the additional revenue is exactly:

    $$m\cdot T \cdot (1-1/T)^n \cdot \left(1-(1-1/T)^x\right) \leq mT(1-1/T)^n\cdot x/T = mx(1-1/T)^n.$$
\end{proof}

\begin{lemma}
The Naive Auction satisfies $\rev^{NA}_n(\ERT^m) \geq \srev_n(\ERT^m) + \Omega(m\sqrt{mn} (1-1/T)^n)$.\footnote{Recall that the Naive Auction is not BIC -- this analysis is just to supply intuition for our later (more involved) computations.}
\end{lemma}
\begin{proof}
    For each item $j$, the Naive Auction achieves revenue $T$ whenever selling separately achieves revenue $T$ (whenever some bidder $i$ has $v_{ij} = T$). The Naive Auction achieves additional revenue in cases where no bidder has value $T$.

    For a fixed item $j$, the probability that all $n$ bidders have value $<T$ is $(1-1/T)^n$. Conditioned on this, we want to find the probability that some bidder both has value at least $mn/T$ for item $j$, and also $T$ for some other item. These are independent events across both bidders and items.

    For a fixed bidder $i$, the probability that $v_{ij} \geq mn/T$ conditioned on $v_{ij} < T$ is exactly $\frac{T/mn - 1/T}{1-1/T} = \frac{\lambda/\sqrt{mn} - 1/(\lambda \sqrt{mn})}{1-1/(\lambda \sqrt{mn})} = \Omega(1/\sqrt{mn}) = \Omega(1/T)$ (as $\lambda > 1$ is an absolute constant). The probability that a fixed bidder $i$ has value $T$ for some item $\neq j$ is simply $1-(1-1/T)^{m-1} = \Omega(m/T)$.\footnote{This follows as $1-(1-1/x)^y = \Omega(y/x)$ when $y < x$, and that $m < T$.} Therefore, the probability that a fixed bidder $i$ has $v_{ij} > mn/T$ \emph{and} $v_{ij'} = T$ for some $j' \neq j$, conditioned on $v_{ij} < T$ is $\Omega(m/T^2)$. 


    Finally, the probability that at least one bidder $i$ has both $v_{ij} \geq mn/T$ and $v_{ij'} = T$ for some $j' \neq j$ is $\Omega(nm/T^2) = \Omega(1)$.\footnote{This again follows from the fact that $1-(1-1/x)^y = \Omega(y/x)$ when $y < x$, and that $n < T^2/m$.} This means that the additional revenue per item gained by the Naive Auction over selling separately is $\Omega((1-1/T)^n \cdot (mn/T)) = \Omega( (1-1/T)^n \cdot \sqrt{mn})$, and multiplying by $m$ items establishes the result.
\end{proof} 

\begin{corollary}\label{cor:naive}
    There exists an absolute constant $c'$ such that $\rev^{NA}_n(\ERT^m) \geq \srev_{n+c'\sqrt{mn}}(\ERT^m)$.
\end{corollary}

Corollary~\ref{cor:naive} establishes that our (not at all BIC) Naive Auction would witness the desired lower bound on $\sscomp_{\mathcal{A}_m}$, if only it were BIC. One obvious problem with the Naive Auction is that it extracts full welfare from buyers with value $T$ for all items, and yet awards items with non-zero probability to buyers with lower values. Our next step is to adjust the payments to address this specific issue (this will still not result in a BIC auction, but it is the first of two steps). 

\subsection{Step Two: A Closer-to-BIC Auction} \label{sec:LNA}
Our next step is to address the obvious issue with the Naive Auction by keeping the same allocation rule with less problematic payments. This will still not yet result in a BIC auction, but will get close to the correct format.

\begin{definition}[The Less-Naive Auction]
The Less-Naive Auction allocates each item separately as follows.
\begin{enumerate}
    \item Use the same \emph{allocation rule} as the Naive Auction. Let $a_0$ denote the interim allocation probability of winning item $j$ conditioned on reporting $v_{ij} = T$, and $b_0$ denote the interim allocation probability of winning item $j$ conditioned on reporting $v_{ij} \in [mn/T, T)$ and $v_{ij'} = T$ for some $j' \neq j$. 
    \item If the bidder $i$ receiving item $j$ has $v_{ij} < T$, charge $mn/T$ (as in the Naive Auction).
    \item If the bidder $i$ receiving item $j$ has $v_{ij} = T$, and also has $v_{ij'} < T$ for all $j' < j$, charge $T$ (as in the Naive Auction).
    \item If the bidder $i$ receiving item $j$ has $v_{ij} = T$, and also has $v_{ij'} = T$ for some $j' < j$,\footnote{Note that we use $j'<j$ to ensure that \emph{exactly} one such item valued at $T$ is \emph{not} subsidized -- the lowest-indexed such item.} then charge a price of $T - \frac{b_0}{a_0}\Paren{T - \frac{mn}{T}}$. Think of $\frac{b_0}{a_0}\Paren{T - \frac{mn}{T}}$ as a subsidy.
    \item Otherwise, do not allocate or elicit payments for item $j$.
    \item Observe that this results in the following possible interim allocations/probabilities for each bidder:
    \begin{itemize}
        \item Receive any single item with interim probability $a_0$, paying interim price $a_0 T$. 
        \item Receive any non-empty set $H$ of items with interim probability $a_0$ and any (possibly empty) set $L$ of items with interim probability $b_0$, paying interim price $b_0 \cdot |L| \cdot \frac{mn}{T} + a_0 \cdot |H| \cdot T - b_0 \left(T-\frac{mn}{T}\right) \cdot (|H|-1)$.
        \item Receive nothing and pay nothing.
    \end{itemize}
\end{enumerate}
\end{definition}

Let us first observe that $a_0$ is small, but not terribly small ($\Theta(T/n) = \Theta(\sqrt{m/n})$ -- roughly the inverse of the expected number of bidders with value $T$ for a single item). $b_0$, on the other hand, is exponentially small -- at most $(1-1/T)^{n-1} = e^{-\Omega(n/T)} = e^{-\Omega(\sqrt{n/m})}$. This initially seems like good news -- even if we return a subsidy on \emph{every} item, the subsidies are exponentially small, and therefore the revenue of the Less-Naive Auction falls short of the Naive Auction by at most an exponentially small amount.

However, recall that the Naive Auction's gains over selling separately are also exponentially small, so this exceptionally simple argument doesn't quite suffice, and these subsidies roughly cancel the gains over selling separately \emph{if we pay them out every time}.\footnote{Essentially, paying the subsidies every time amounts to selling each item $j$ using a randomized-but-still-single-dimensional auction, which again cannot outperform selling separately.} However, there is one key case where we don't need to pay a subsidy: if bidder $i$ values exactly one item at $T$. Indeed, Lemma \ref{lemma:LN-vs-SREV} formalizes this intuition and shows that the Less-Naive Auction's gain in revenue comes precisely from selling items $\neq j$ for cheap to bidders who value a single item $j$ at $T$.



\begin{lemma}\label{lemma:LN-vs-SREV}
The Less-Naive Auction satisfies:\footnote{Again recall that the Less-Naive Auction is not BIC -- this analysis is just to supply intuition for our later (even more involved) computations.}
\[
    \rev_n^{LNA}(\ERT^m) = \srev_n(\ERT^m)+b_0mn^2 \Paren{\frac{T}{mn} - \frac{1}{T}} \Paren{1 - \Paren{1 - \frac{1}{T}}^m - \frac{m}{T}\Paren{1 - \frac{1}{T}}^{m-1}}.
\]
\end{lemma}

\begin{proof}
The revenue extracted by the Less-Naive Auction from selling item $j$ to bidder $i$ is
\begin{align*}
    a_0 T \cdot \1\Paren{v_{ij} = T, \max_{j' < j} v_{ij'} < T} & + a_0 \Paren{T - \frac{b_0}{a_0}\Paren{T - \frac{mn}{T}}} \cdot \1\Paren{v_{ij} = T, \max_{j' < j} v_{ij'} = T} \\
        &+ \frac{b_0 mn}{T} \cdot \1\Paren{v_{ij} \in [\tfrac{mn}{T}, T), \max_{j'\ne j} v_{ij'} = T}
\end{align*}
Meanwhile, $\srev$ extracts $a_0 T \cdot \1(v_{ij} = T)$.
Thus, the Less-Naive Auction obtains
\[
    \frac{b_0 mn}{T} \cdot \1\Paren{v_{ij} \in [\tfrac{mn}{T}, T), \max_{j'\ne j} v_{ij'} = T} - b_0\Paren{T - \frac{mn}{T}} \cdot \1\Paren{v_{ij} = T, \max_{j' < j} v_{ij'} = T}
\]
more revenue from selling item $j$ to bidder $i$ than $\textsc{SRev}_n$ does.
Across all items, the Less-Naive Auction obtains in expectation
\begin{align*}
    &\frac{b_0 mn}{T} \sum_{j\in [m]} \PP\SqBr{v_{ij} \in [\tfrac{mn}{T}, T), \max_{j'\ne j} v_{ij'} = T } - b_0\Paren{T - \frac{mn}{T}} \sum_{j \in [m]} \PP\SqBr{v_{ij} = T, \max_{j' < j} v_{ij'} = T} \\
    &= \frac{b_0 mn}{T} \sum_{j \in m} \Paren{\frac{T}{mn} - \frac{1}{T}}\Paren{1 - \Paren{1 - \frac{1}{T}}^{m-1}} - b_0\Paren{T - \frac{mn}{T}} \sum_{j\in [m]} \frac{1}{T}\Paren{1 - \Paren{1 - \frac{1}{T}}^{j-1}} \\
    &= \frac{b_0 m^2n}{T} \Paren{\frac{T}{mn} - \frac{1}{T}}\Paren{1 - \Paren{1 - \frac{1}{T}}^{m-1}} - b_0\Paren{T - \frac{mn}{T}} \underbrace{\Paren{\frac{m}{T} - 1 + \Paren{1 - \frac{1}{T}}^m}}_{\text{expected number of subsidies}} \tag{geometric sum} \\
    &= b_0mn \Paren{\frac{T}{mn} - \frac{1}{T}} \Paren{1 - \Paren{1 - \frac{1}{T}}^m - \frac{m}{T}\Paren{1 - \frac{1}{T}}^{m-1}}
\end{align*}
more revenue from bidder $i$ than $\textsc{SRev}_n$ does.
Summing over all bidders yields the lemma.
\end{proof}

We highlight that since $m \leq n$, the expected number of subsidies per bidder is strictly less than 1. In fact, when $m \ll n$, we expect to pay almost no subsidies. Thus, we expect our extra revenue to come from selling items for which no bidders have value $T$ to bidders who value exactly one other item at $T$.

Finally, let us revisit incentives of the Less-Naive Auction. The Less-Naive Auction is \emph{almost} BIC. Indeed, any bidder who values at least one item at $T$ is incentivized to report truthfully (we will prove this formally in the subsequent Section~\ref{sec:menu-format}). Moreover, any bidder who values all items far from $T$ will also prefer to take no items and pay nothing. However, a bidder with $v_{ij}$ \emph{extremely close} to $T$ for item $j$ and $v_{ij'}> mn/T$ for another (and $v_{ij''} < T$ for all $j''$) \emph{may} prefer to misreport that $v_{ij} = T$ (taking a small negative utility on item $j$) in order to receive item $j'$ with non-zero probability. Therefore, the Less-Naive Auction is not BIC.

However, because $b_0$ is exponentially small, the maximum possible gain of such a misreport is also exponentially-small. Therefore, only types with $v_{ij}$ \emph{inverse exponentially-close} to $T$ will even consider this misreport, and there is hope that a slight modification to the Less-Naive Auction might work. Indeed, we now show an auction whose interim menu takes the same format as the Less-Naive Auction, but with parameters $a \approx a_0$ and $b \approx b_0$ that is BIC and again has essentially the same expected revenue.

\subsection{Step Three: A BIC Auction}\label{sec:menu-format}

Now, we introduce a final set of modifications to make the Less-Naive Auction \emph{fully} BIC. Recall that the Less-Naive Auction is not BIC because a bidder with no values equal to $T$ but sufficiently high $v_{ij} \approx T$, $v_{ij'} > mn/T$ may choose to misreport and take a small loss on item $j$ in order to receive item $j'$ (indeed, Corollary \ref{corollary:preferred-option} formalizes this intuition). It follows that the interim allocation probabilities $a_0, b_0$ are not actually feasible, because there are more bidders that want to purchase items than we can keep our promises to.

To fix this, we maintain the same menu format, but lower the interim probabilities (and prices accordingly) to $a, b$. This fixes the incentive issues of the Less-Naive Auction by making misreporting less attractive, but since $a \approx a_0$ and $b\approx b_0$ we still attain approximately the same revenue. Observe that even a small change in $a_0$ and $b_0$ works, because of two simultaneous effects at play: (1) lowering the allocation probabilities inherently increases the number of bidders to which it is feasible to allocate an item; (2) lowering $b_0$ in particular reduces the set of types that may prefer to misreport. Balancing these two effects so that the feasibility constraint is tight (we would like to allocate the item as much as possible, so that we can extract as much revenue as possible) results in a system with a fixed point $(a,b)$ that is not too far from the original $(a_0, b_0)$.

\begin{definition}[The Not-So-Naive Auction]
The Not-So-Naive Auction allocates the items according to the following menu of interim allocations/probabilities for each bidder:
\begin{itemize}
        \item Receive any single item with interim probability $a$, paying interim price $a T$. 
        \item Receive any non-empty set $H$ of items with interim probability $a$ and any (possibly empty) set $L$ of items with interim probability $b$, paying interim price $b \cdot |L| \cdot \frac{mn}{T} + a \cdot |H| \cdot T - b \left(T-\frac{mn}{T}\right) \cdot (|H|-1)$.
        \item Receive nothing and pay nothing.
\end{itemize}
\end{definition}

Again, we highlight that if we had $a = a_0, b=b_0$, this menu just describes the Less-Naive Auction. In Lemma~\ref{lemma:preferred-option} and Corollary~\ref{corollary:preferred-option}, we characterize the incentive properties of \emph{any} such menu of the above form parametrized by $a\ge b$; following this, we proceed to set $a$ and $b$ specifically so that the Not-So-Naive Auction is feasible.

\begin{lemma}\label{lemma:preferred-option}
Suppose $T \geq \sqrt{mn}$ and $a \geq b$.
Let $v \in [1, T]^m$ and let $j^* \in \arg\max_j v_j$.
Define $H := \CrBr{j : v_j = T} \cup \{j^*\}$ and $L := \CrBr{j : v_j \geq mn/T} \setminus H$.
For all $H' \in 2^{[m]} \setminus \{\varnothing\}$ and $L' \subseteq [m] \setminus H'$, a bidder with type $v$ prefers the menu option $(H,L)$ over the menu option $(H', L')$.
\end{lemma}

Lemma \ref{lemma:preferred-option} says that a bidder with type $v$ prefers $(H,L)$ over any other option that allocates an item with some positive probability.
In particular, it does \emph{not} say whether a bidder with type $v$ would prefer $(H,L)$ over not getting any items at all.

\begin{proof}
The utility of a bidder with type $v$ for the menu option $(H', L')$ is
\begin{align*}
    a \sum_{j \in H'} v_j + b \sum_{j \in L'} v_j & - \Paren{\abs{H'}aT + \abs{L'} b \frac{mn}{T} - (\abs{H'} - 1) b \Paren{T - \frac{mn}{T}}} \\
        &= \sum_{j \in H'} \Paren{a v_j - \Paren{(a-b) T + b \frac{mn}{T}}} + \sum_{j \in L'} b \Paren{v_j - \frac{mn}{T}} - b\Paren{T - \frac{mn}{T}}.
\end{align*}
We show that $(H,L)$ maximizes this utility over all $H' \in 2^{[m]} \setminus \{\varnothing\}$ and $L' \subseteq [m] \setminus H'$.

Consider the difference in utility of getting item $j$ with probability $a$ and getting the same item with probability $b$:
\[
    av_j - \Paren{(a-b)T + b\frac{mn}{T}} - b \Paren{v_j - \frac{mn}{T}} = (a-b)(v_j - T).
\]
It is clear that if $a \geq b$, then there are only two cases in which a bidder who values item $j$ at $v_j$ (and who is forced to get at least one item) would prefer to get item $j$ with probability $a$: either (1) $v_j = T$, or (2) she does not value any item at $T$, but $j = j^*$ (if she must pay $T$ for \emph{some} item, her utility is least negative when $j$ is the item she values the most rather than some other item). 
Thus, the only items that a bidder with type $v$ prefers to get with probability $a$ rather than $b$ are the items in $H$.
Getting any item outside of $H$ with probability $a$ strictly decreases utility.

Of the items not in $H$, a bidder with type $v$ would only choose to get those with values at least $mn/T$ with some positive probability, since all options cost at least $mn/T$, so the bidder would be overpaying for any item valued less than $mn/T$ (which strictly decreases utility). Note that the items with values at least $mn/T$ that are not in $H$ are precisely those in $L$.

As discussed before, the only items that a bidder with type $v$ prefers to get with probability $a$ rather than $b$ are the items in $H$, so such a bidder prefers to get the items in $L$ with probability $b$ instead of $a$. Thus, the bidder's most preferred menu item is exactly $(H,L)$.
\end{proof}

\begin{corollary}\label{corollary:preferred-option}
Suppose $T \geq \sqrt{mn}$ and $a \geq b$.
Let $v \in [1, T]^m$ and let $j^* \in \arg\max_j v_j$.
Define $H := \CrBr{j : v_j = T} \cup \{j^*\}$ and $L := \CrBr{j : v_j \geq mn/T} \setminus H$.
\begin{itemize}
    \item A bidder who values some item at $T$ prefers the menu option $(H, L)$ over any other option in the menu (including not receiving any items).
    \item A bidder who does not value any item at $T$ prefers the menu option $(\{j^*\}, L)$ over any other option in the menu (including not receiving any items) if and only if $av_{j^*} + b \sum_{j \in L} v_j \geq aT + \abs{L}b\frac{mn}{T}$.
\end{itemize}
\end{corollary}

Now, we define $a$ and $b$ such that the resulting menu is feasible.
We do so implicitly.
Recall that setting $a = a_0$ and $b = b_0$ is infeasible because there are bidders who do not value any items at $T$ yet value a subset of the items enough to be willing to pay $T$ for an item in order to be eligible to get additional items at lower prices.
The probability that there exists a bidder who values each item less than $T$ yet is willing to purchase the menu option $(\{j^*\}, L)$ is 
\[
    q_{\ell} \coloneqq \underset{v}{\PP}\SqBr{v_{j^*} = \max_j v_j < T, \min_{j \in \{j^*\} \cup L} v_j \geq \frac{mn}{T}, \max_{j \not\in \{j^*\} \cup L} v_j < \frac{mn}{T}, av_{j^*} + b \sum_{j \in L} v_j \geq aT + \abs{L}b\frac{mn}{T}}. 
\]
Note that for a given $\ell \geq 1$, the above probability is the same for any choice of $j^* \in [m]$ and $L \subseteq [m] \setminus \{j^*\}$ such that $\abs{L} = \ell$, so we may denote it by $q_\ell$.

Since there are more bidders who want to purchase items than just those with $T$ values, the interim allocation probabilities $a$ and $b$ must be smaller than $a_0$ and $b_0$ to accommodate these bidders.
More specifically, if we term bidders who are willing to receive item $j$ with probability $a$ as ``high'' and those who are only willing to receive item $j$ with probability $b$ as ``low,'' and we want to allocate each item uniformly at random to the high bidders before allocating uniformly at random to the low bidders, then $a$ and $b$ must satisfy the following implicit definitions.
\begin{align*}
    a &= \underset{v_{-i}}{\EE}\SqBr{\frac{1}{1 + \sum_{k \not= i} \1(\text{$i$ high})}} 
        = \sum_{k=0}^{n-1} \frac{\binom{n-1}{k}}{k+1} \PP\SqBr{\text{high}}^k \PP\SqBr{\text{not high}}^{n-k-1} 
        = \frac{1 - \Paren{1 - \PP\SqBr{\text{high}}}^n}{n\PP\SqBr{\text{high}}}, \\
    b &= \underset{v_{-i}}{\EE}\SqBr{\frac{\1(\text{$\not\exists$ high})}{1 + \sum_{k \not= i} \1(\text{$i$ low})}} 
        = \frac{\PP_{v_{-i}}\SqBr{\text{$\not\exists$ high}} \Paren{1 - \Paren{1 - \PP\cSqBr{\text{low}}{\text{not high}}}^n}}{n\PP\cSqBr{\text{low}}{\text{not high}}}, \tag{by bidder independence, $\PP\cSqBr{\text{low}}{\text{$\not\exists$ high}} = \PP\cSqBr{\text{low}}{\text{not high}}$}
\end{align*}
where 
\begin{align*}
    \PP\SqBr{\text{high}} 
        &= \textstyle \underset{v}{\PP}\SqBr{v_j = T} + \sum\limits_{L \subseteq [m] \setminus \{j\}} q_{\abs{L}} 
        = \frac{1}{T} + \sum\limits_{\ell = 1}^{m-1} \binom{m-1}{\ell} q_\ell, \\
    \PP\SqBr{\text{low}} 
        &= \textstyle \underset{v}{\PP}\SqBr{\substack{\max\limits_{j' \not= j} v_{j'} = T, \\ \frac{mn}{T} \leq v_j < T}} + \sum\limits_{j^* \not= j} \sum\limits_{\substack{L \subseteq [m] \setminus \{j^*\} : \\ j \in L}} q_{\abs{L}} 
        = \Paren{1 - \Paren{1 - \frac{1}{T}}^{m - 1}} \Paren{\frac{T}{mn} - \frac{1}{T}} + (m-1) \sum\limits_{\ell = 1}^{m-1} \binom{m-2}{\ell - 1} q_{\ell}.
\end{align*}
The expression for $\PP\SqBr{\text{high}}$ follows from the fact that by Corollary~\ref{corollary:preferred-option}, a bidder is high for an item if and only if (1) she has value $T$ for it or (2) it is her favorite item and she has sufficiently high values for the other items.
Similarly, a bidder is low for an item if and only if her value for it is in $[mn/T, T)$ and either (1) she has value $T$ for some other item or (2) she has sufficiently high values for the other items.
We point out that our definitions of $a$ and $b$ are indeed implicit since $q_1, \dots, q_{m-1}$ depend on $a$ and $b$.

We now show that $a \geq b$, so that Lemma \ref{lemma:preferred-option} holds with our definitions of $a$ and $b$.

\begin{restatable}{lemma}{aVsB}\label{lemma:a-geq-b}
If $T \geq \sqrt{mn}$, then $\frac{b}{a} \leq \frac{n}{T} e^{-\frac{n}{T}} \Paren{1 - e^{-\frac{n}{T}}}^{-1}$.
\end{restatable}

\begin{corollary}
If $T \geq \sqrt{mn}$, then $a \geq b$.
\end{corollary}

\begin{proof}
A direct consequence of Lemma \ref{lemma:a-geq-b} and the fact that $xe^{-x} \leq 1 - e^{-x} $ for $x \geq 0$.
\end{proof}

\subsection{Step Four: Comparing the Revenue of the Not-So-Naive Auction to the Revenue of the Less-Naive Auction}

The ultimate goal is to compare the revenue of the Not-So-Naive Auction against the revenue of selling separately. We proceed via an intermediate comparison between the Not-So-Naive Auction and the Less-Naive Auction (which we have already compared to $\srev_n(\ERT^m)$ in Section~\ref{sec:LNA}).

\mattnote{What is $q_\ell$ below? I can't find where it's defined in the body -- is it a term used in teh calculation proofs? This is also used in Section~3.5 below.}

\begin{restatable}{lemma}{NSNvsLN}\label{lemma:NSN-vs-LN}
If bidders report their values truthfully in the Less-Naive Auction, then $\rev_n^{NSN} (\ERT^m)$ exceeds $\rev_n^{LNA} (\ERT^m)$ by at least
\[
    (b-b_0)mn^2 \Paren{\frac{T}{mn} - \frac{1}{T}} \Paren{1 - \Paren{1 - \frac{1}{T}}^m - \frac{m}{T}\Paren{1 - \frac{1}{T}}^{m-1}} + \frac{bm^2(m-1)n^2}{T} \sum\limits_{\ell = 1}^{m-1} \binom{m - 2}{\ell-1} q_\ell .
\]
\end{restatable}

\begin{corollary}\label{corollary:NSN-vs-SREV}
The Not-So-Naive Auction satisfies: 
\begin{equation*}
    \begin{split}
        &\rev^{NSN}_n(\ERT^m) \ge  \srev_n(\ERT^m) \\
        & + bmn^2 \Paren{\Paren{\frac{T}{mn} - \frac{1}{T}} \Paren{1 - \Paren{1 - \frac{1}{T}}^m - \frac{m}{T}\Paren{1 - \frac{1}{T}}^{m-1}} + \frac{m(m-1)}{T} \sum_{\ell = 1}^{m-1} \binom{m - 2}{\ell-1} q_\ell}.
    \end{split}
\end{equation*}

\end{corollary}

\begin{proof}
A direct consequence of Lemmas \ref{lemma:LN-vs-SREV} and \ref{lemma:NSN-vs-LN}.
\end{proof}

\subsection{Step Five: Bounding $\comp_{\ERT^m}(n)$}






Notice that $\comp_{\ERT^m}(n)$ is \emph{at least} the smallest $c$ such that $\srev_{n+c}(\ERT^m)$ exceeds $\rev_n^{NSN}(\ERT^m)$. Combining Lemma \ref{lem:srev} and Corollary \ref{corollary:NSN-vs-SREV}  along with the fact that all $q_\ell \ge 0$, this occurs only if\footnote{We ignore the expected revenue gained from selling items to ``low''-type bidders with no $T$ values since this term is irrelevant to the remainder of our analysis.}
\[
    c  \geq \frac{bn^2 \Paren{\frac{T}{mn} - \frac{1}{T}} \Paren{1 - \Paren{1 - \frac{1}{T}}^m - \frac{m}{T}\Paren{1 - \frac{1}{T}}^{m-1}}}{\Paren{1 - \frac{1}{T}}^n}.
\]


By the union bound, we expect the RHS to behave like $b mn / (T (1 - \frac{1}{T})^n)$.
If we set $T \sim \sqrt{mn}$, then to show that the competition complexity is $c = \Omega(\sqrt{mn})$, it suffices to show that $b = \Omega\Paren{\Paren{1- \frac{1}{T}}^n}$.
Intuitively, if the probability of the types of bidders for which we had to modify the Less-Naive Auction to the Not-So-Naive Auction is sufficiently small, then $b \approx b_0 = \Omega((1 - \frac{1}{T})^n)$ since the probability of valuing item $j$ in $[mn/T, T)$ and some other item at $T$ is at most $O(1/n)$ if $T \sim \sqrt{mn}$. 
We show that all of this is indeed the case in Lemma \ref{lemma:b-lower-bound} and prove that the competition complexity is $\Omega(\sqrt{mn})$ in Theorem \ref{theorem:competition-complexity-sqrt-mn}.




\begin{restatable}{lemma}{bLB}\label{lemma:b-lower-bound}
If $T = \lambda \sqrt{mn}$ for some constant $\lambda > 1$, then $b = \Omega\Paren{\Paren{1 - \frac{1}{T}}^n}$.
\end{restatable}

\begin{restatable}{theorem}{mainResult}\label{theorem:competition-complexity-sqrt-mn}
     If $T = \lambda \sqrt{mn}$ for some constant $\lambda > 1$, then $\comp_{\ERT^m}(n) = \Omega(\sqrt{mn})$.
\end{restatable}

This concludes the proof of our main result. We have explicitly defined a BIC auction (the Less-Naive Auction) for $(\ERT^m)^n$ whose revenue exceeds $\srev_{\ERT^m}(n+c \sqrt{nm})$ for some absolute constant $c > 0$, and all $m \geq 2$, and all $n$.

\section{Further Preliminaries: Dual Flow Benchmarks} \label{sec:addprelims}

In the following sections, we provide revenue upper bounds using the~\cite{CaiDW16} framework. We briefly state the minimal preliminaries necessary to get started, and will fully flesh out terminology as needed during proofs. As needed, we will clarify what is a `useful dual flow' so that the statement of Theorem~\ref{thm:benchmark_thm6} is fully self-contained. Theorem~\ref{thm:canonical_benchmark} is already self-contained. Their framework establishes the following revenue benchmark in terms of the induced \emph{virtual values} of bidder $i$ for item $j$, $\Phi_{ij}^\lambda(\vec{v}_i)$ as a function of a `useful dual flow' $\lambda$:

\begin{theorem}[\cite{CaiDW16}, Theorem 6] \label{thm:benchmark_thm6}
    Let $\lambda$ be any useful dual flow, and  $M=(\pi, p)$ be a BIC mechanism. The revenue of $M$ is less than or equal to the virtual welfare of $\pi$ with respect to the virtual value function $\Phi^\lambda$; that is:
        $$ \rev_n^M(D) \le \sum_{i=1}^n \sum_{j=1}^m \EE_{\vec{v} \leftarrow D^n} \left[ \pi_{ij}(\vec{v}_i) \cdot \Phi_{ij}^\lambda(\vec{v}_i) \right].$$
    
\end{theorem}


Taking the supremum over all feasible $M$ then provides an upper bound on $\rev_n(D)$. In particular, when $\lambda$ is the \emph{canonical flow} that divides the type space into regions $R_j$ based on the favorite item $j$, then uses a Myersonian-like flow within each region (refer to~\cite{CaiDW16} for a precise definition), the following \emph{relaxation} of the benchmark is useful:

\begin{theorem}[\cite{CaiDW16}, Corollary 28] \label{thm:canonical_benchmark} Let $R_j:=\{\vec{v} \ |\ \arg\max_{\ell\in[m]}\{v_\ell\}=j\}$ (with ties broken lexicographically). Then 
    $$ \rev_n(D) \le \sum_{j=1}^m \EE_{\vec{v} \leftarrow D^n} \left[ \max_{i\in [n]} \left\{ \bar{\varphi}_j (v_{ij}) \cdot \1\Paren{\vec{v}_i \in R_j} + v_{ij} \cdot \1\Paren{\vec{v}_i \notin R_j}  \right\} \right].  $$
\end{theorem}

\section{A Reduction from $\mathcal{A}_m^R$ to $\mathcal{ERC}_m$ via Stochastic Dominance} \label{sec:ERTreduction}

We now consider the competition complexity of an arbitrary distribution $D  \in \mathcal{A}_m^R$. The key idea is that the Theorem~\ref{thm:canonical_benchmark} benchmark can be written entirely in terms of (ironed) virtual values, which then allows a direct comparison to the virtual value obtained by $\srev$. If $\srev$ (with $c$ additional bidders) always obtains a virtual value of a higher quantile than the optimal mechanism (without additional bidders), then it also achieves higher revenue. To tighten the previous analysis from \cite{CaiDW16, BeyhaghiW19}, we first use the style of Theorem~\ref{thm:canonical_benchmark} but \emph{applied to a specific allocation}, then take the supremum over all feasible and BIC mechanisms. 

Fix a deterministic allocation rule $x$, and let $x_j(\vec{v})$ denote the winner of item $j$ under $x$ on input $\vec{v}$ (if no one wins the item, let $x_j(\vec{v}) = \bot$). Sample $\vec{v} \leftarrow D^n$, and choose a bidder $i$ for item $j$ according to $x$. Define the following quantities: 

\begin{itemize}
    \item $CDW_j^x (\vec{v}) \coloneqq \sum_{i} x_{ij}(\vec{v}) \cdot \Phi_{ij}^\lambda (\vec{v}_i) = \sum_i x_{ij}(\vec{v}) \cdot (\bar{\varphi}_j (v_{ij}) \cdot \1\Paren{\vec{v}_i \in R_j} + v_{ij} \cdot \1\Paren{\vec{v}_i \notin R_j})$, the virtual value (using the canonical flow) of the recipient of item $j$ on valuation profile $\vec{v}$,

    \item $Q_j^x (\vec{v})$, a random variable sampled as follows: 
        \begin{itemize}
            \item If $\vec{v}_i \in R_j$, output $Q_j^x = q_j(\vec{v}_i) = F_j(v_{ij})$.
            \item Else, output $Q_j^x \leftarrow U[q_j(\vec{v}_i), 1]$.
        \end{itemize}

    \item $S_{n+c} \coloneqq \max_{\vec{q} \leftarrow U[0,1]^{n+c}} q_i$, the maximum of $n+c$ independently drawn quantiles.
\end{itemize}

Following the rest of the argument from \cite{CaiDW16, BeyhaghiW19} gives the following \emph{refined} benchmark (note that this essentially interchanges the expectation and the maximum from Theorem~\ref{thm:canonical_benchmark}, so it indeed furnishes a tighter bound):
$$ \rev_n(D) \le \sup_{\text{feasible BIC } x} \sum_{j=1}^m \EE_{\vec{v} \leftarrow D^n} \left[ CDW_j^x (\vec{v}) \right].$$

Using this improved benchmark and the language of quantile space, we now compare $\rev_n(D)$ to $\srev_{n+c}(D)$. Along the way, we appeal to the specific form of the virtual values of distributions in $\mathcal{ERC}_m$, which establishes the non-trivial direction of the equality between $\comp_{\mathcal{A}_m^R}(n)$ and $\comp_{\mathcal{ERC}_m}(n)$. Full proofs of all results are provided in Appendix~\ref{sec:proofs_ERTreduction}.

\begin{restatable}{proposition}{redclaimone} \label{prop:reduction_claim1}
    For all allocation rules $x$ and all items $j$, 
    $$ \EE_{\substack{\vec{v}\leftarrow D^n, \\r}} \left[ \bar{\varphi}_j \left( F_j^{-1} \left(Q_j^x (\vec{v}) \right) \right) \right] \ge \EE_{\vec{v}\leftarrow D^n} \left[ CDW_j^x (\vec{v}) \right]. $$
\end{restatable}

\begin{restatable}{proposition}{redclaimtwo}\label{prop:reduction_claim2}
    If $S_{n+c} \succsim Q_j^x(\vec{v})$, then $\srev_{n+c}(D) \ge \EE_{\vec{v} \leftarrow D^n} \left[\sum_j CDW_j^x (\vec{v}) \right]$.
\end{restatable}

\begin{observation}\label{obs:ERTquantiles}
    For all $q \in [0,1)$, the distribution satisfying (up to scaling by a constant) $\bar{\varphi}(x) = \1\Paren{F(x) \ge q}$ is $\ERT[\frac{1}{1-q}]$.
\end{observation}

\begin{restatable}{proposition}{redclaimthree}\label{prop:reduction_claim3}
    $S_{n+c} \succsim Q_j^x(\vec{v})$ if and only if $$\srev_{n+c}(\ERT) \ge \EE_{\vec{v} \leftarrow (\ERT^n}) \left[\sum_j CDW_j^x (\vec{v}) \right]$$ for all truncations $T \in [1, \infty)$.
\end{restatable}

\begin{restatable}{corollary}{reductioncc}  \label{cor:reduction_cc}
    If $\srev_{n+c}(\ERT) \ge \EE_{\vec{v} \leftarrow (\ERT)^n} \left[\sum_j CDW_j^x (\vec{v}) \right]$ for all $T \in [1, \infty)$, then the competition complexity for \emph{any} distribution $D \in \mathcal{A}_m^R$ is $O(c)$.
\end{restatable}

This reduction implies that to establish a bound on $\comp_{\mathcal{A}_m^R}(n)$, it suffices to just study $\comp_{\mathcal{ERC}_m}(n)$. Although this claim is not directly used for our earlier main result, we still present it as a technique of general interest, potentially useful for future work, and illuminating as to the context and further implications of our main result focusing on $\ERT$. 
\section{Upper Bound on $\rev_n(\ER^m) = nm+ O(m^2\ln n)$ when $n > m$} \label{sec:upperbound}


In this section, we show that we cannot obtain more than $nm + O(m^2 \ln n)$ in revenue from $n$ bidders with additive valuations for $m$ items drawn iid from $\ER$.
This upper bound is interesting for two reasons.
First, it shows that the untruncated Equal Revenue curve does \textit{not} witness the worst-case competition complexity when $n \geq m$.
However, in Section \ref{sec:kfbic}, we show that if bidders cannot lie about their favorite item, then the untruncated Equal Revenue curve \textit{does} witness the worst-case competition complexity when $n \geq m$.
Thus, a ``region-separated'' flow provably cannot give a tight upper bound on the revenue obtainable for this setting.
In Sections \ref{sec:logn_m=2} and \ref{sec:ER-revenue-UB}, we demonstrate how to circumvent this impossibility by taking advantage of certain cross-region constraints.
We show that our upper bound is nearly tight in Section \ref{sec:ER-revenue-LB}.

\subsection{Tight Bound for $m=2$: $\rev_n(\ER^2) = 2n+ \Theta(\ln n)$} \label{sec:logn_m=2}

To establish a tight bound on $\rev_n(\ER^2)$, we start from Theorem~\ref{thm:benchmark_thm6}. Rather than relaxing all the way to Theorem~\ref{thm:canonical_benchmark}, we obtain an upper bound on the optimal revenue by first providing a further characterization of feasible and BIC mechanisms $M$ over which we take the supremum of the virtual welfare.

As established in~\cite{CaiDW16}, the expected virtual welfare from bidders who are awarded their favorite item is $2n$; we seek to understand the expected virtual welfare from bidders who win their non-favorite item. 
We begin with some motivating observations (proofs of which are provided in Appendix~\ref{sec:proofs_logn_m=2}):

\begin{restatable}{observation}{ERmotivationone}\label{obs:ERmotivation1}
    Fix $v_1 > v_2$. It is feasible to have  each type $(v_1, v_2)$ with $v_2 \ge \ln^2 n$ receive their non-favorite item (item 2) with probability $\min\left\{\frac{1}{2}, \frac{v_1}{10n}\right\}$.
\end{restatable}

\begin{restatable}{observation}{ERmotivationtwo}\label{obs:ERmotivation2}
    If every type with $v_N \ge \ln^2 n$ receives their non-favorite item with probability at least $\min \left\{\frac{1}{2}, \frac{v_F}{10n} \right\}$, the benchmark gets at least $2n + \Omega(\ln^2 n)$.
\end{restatable}

\begin{restatable}{observation}{ERmotivationthree}\label{obs:ERmotivation3}
    It is \emph{not} feasible to allocate both items with probability $\min \left\{\frac{1}{2} , \frac{v_1}{10n} \right\}$ to \emph{all} types $(v_1, v_2)$. 
\end{restatable}

Combining Observations~\ref{obs:ERmotivation1} and~\ref{obs:ERmotivation3} suggests that we are only in trouble if it is somehow possible to give items only to players with big $v_N$ without also giving items to players with small $v_N$. But, this is difficult if items are mostly awarded based on having large $v_F$ (because $v_N$ is generally much smaller than $v_F$ when $v_F$ is large). That is, in order to get expected virtual welfare $2n + \Omega(\ln^2 n)$, we need to have things like ``$(n/100, \sqrt{n})$ gets item 2 with probability $1/100$, but $(n/10, 2)$ gets item 1 with probability $0$.'' So, our goal is to show that this is \emph{not} possible. We do this by appealing to BIC and IR constraints in Lemma~\ref{lem:MattLemma} in addition to feasibility in Corollary~\ref{cor:MattCor}; see Appendix~\ref{sec:proofs_logn_m=2} for full details.

\begin{restatable}{lemma}{mattlem}\label{lem:MattLemma}
    Let $(v_1, v_2)$ get item 2 with probability $q \coloneqq \pi_2(v_1, v_2)$. Let also $v_2' \leq v_2$. Then $(3v_1, v_2')$ gets item 1 with probability $\pi_1 (3v_1, v_2') \ge q/4$. 
\end{restatable}

\begin{restatable}{corollary}{mattcor}\label{cor:MattCor}
    Let $y \geq 2$. Then $\EE_{v_1} \left[\pi_2(v_1, v_2) \mid v_2 = y \right] = \int_{y}^\infty \pi_2(v_1,y) \cdot \frac{y}{v_1^2} dv_1 \leq \frac{24y}{n}$. That is, the probability of getting item 2 conditioned on having $v_2 = y$ is at most $\frac{24y}{n} = O(\frac{y}{n})$.
\end{restatable}

To wrap up:
\begin{lemma}
    The expected contribution to the virtual welfare from the non-favorite item is $O(\ln n)$.
\end{lemma}
\begin{proof}
    First, note that a bidder only contributes to the virtual welfare if they are awarded the item (which occurs with probability $q_2 = \max\{O(\frac{y}{n}), 1\}$, in which case they contribute their virtual value (which is at most $v_2$). Then, we get the following upper bound, which we can compute via the law of total expectation:

\begin{align*}
    \EE_{v_1, v_2} \left[v_2 q_2 \right] &= \EE_{v_2} \left[\EE_{v_1} \left[v_2 q_2 \mid v_2 \right] \right] \\
    &= O\left(\EE_{v_2} \left[v_2 \cdot \max\left\{\frac{v_2}{n}, 1\right\} \right] \right) \\
    &= O\left(\int_{x = 1}^n \frac{x^2}{n} \mathbb{P}(v_2=x) dx + \int_{x = n}^\infty x \cdot \mathbb{P}(v_2=x) dx \right) \\
    &= O\left(\int_{x = 1}^n \frac{x^2}{n} \cdot \frac{2}{x^3} dx + \int_{x = n}^\infty \frac{2}{x^2} dx \right) \\
    &= O\left(\frac{2\ln n}{n} + \frac{2}{n} \right) \\
    &= O\left(\frac{\ln n}{n}\right).
\end{align*}
Summing over all $n$ bidders gives a total of $O(\ln n)$.
\end{proof}

Combining this with the lower bound of $\rev_n(\ER^2) = 2n + \Omega(\ln n)$ due to~\cite{BeyhaghiW19} establishes that this bound is tight.
\begin{theorem}
    $\rev_n(\ER^2) = 2n + \Theta(\ln n)$.
\end{theorem}

\subsection{Generalizing to $m>2$: $\rev_n(\ER^m) = nm+ O(m^2 \ln n)$} \label{sec:ER-revenue-UB}

We generalize the analysis from Section~\ref{sec:logn_m=2} to general $m$, thereby improving the upper bound on the competition complexity of $n$ bidders with additive values drawn i.i.d. from $\ER^m$ to $O(m \ln n)$.
Throughout this section, we introduce the following additional notation when considering a particular bidder:
\begin{itemize}
    \item $E_j$ denotes the event that item $j$ is the favorite item.
    \item $E_{-j}$ denotes the event that item $j$ is a non-favorite item.
\end{itemize}
Our approach exactly mirrors that of the $m=2$ analysis but requires much more involved calculations; as such, we defer all proofs to Appendix~\ref{sec:proofs_ER-revenue-UB}.

\begin{restatable}{lemma}{refinedMatt}\label{lem:refined-MattLemma_m}
    Let $(\pi, p)$ be a strongly monotone BIC mechanism for an arbitrary distribution, and let $v_k \geq v_j \geq v_\ell'$.
For any $\lambda > 1$, $\pi_k(\lambda v_k, v_\ell', v_{-\{k, \ell\}}) \geq \frac{\lambda - 2}{\lambda - 1} \pi_\ell(v)$.
\end{restatable}

\begin{restatable}{lemma}{jointpdfcond}\label{lem:joint-pdf-given-non-fav}
    Let $v \sim \ER^m$.
Conditioned on item $k$ being the favorite item and on the value of item $\ell \not= k$, the distribution of $v_{-\ell}$ satisfies 
\[
    \PP\cSqBr{v_j \leq w_j \; \forall j \not= \ell}{v_\ell, E_k} = \int_{v_\ell}^{w_k} \prod_{j \not= k, \ell} \min\CrBr{\frac{1 - \frac{1}{w_j}}{1 - \frac{1}{v}}, 1} \frac{(m-1)\Paren{1-\frac{1}{v}}^{m-2} \frac{1}{v^2}}{1-\Paren{1-\frac{1}{v_\ell}}^{m-1}}\,\mathrm{d}v.
\]
Thus, the joint pdf of $v_{-\ell}$ conditioned on the value $v_\ell$ and the event $E_k$ is
\[
    \1\Paren{w_k \geq \max\CrBr{v_\ell, \max_{j \not= k, \ell} w_j}} \frac{m-1}{1-\Paren{1 - \frac{1}{v_\ell}}^{m-1}} \prod_{j \not= \ell} \frac{1}{w_j^2}.
\]
\end{restatable}

\begin{restatable}{lemma}{expNFprob}\label{lem:expected-prob-of-getting-non-fav-via-expected-prob-of-getting-particular-fav-improved}
    Let $(\pi, p)$ be a BIC mechanism for $n$ bidders with $m$ additive valuations drawn i.i.d. from $\ER$.
Conditioned on item $k$ being the favorite and on the value of item $\ell \not= k$, the expected interim probability of receiving item $\ell$ satisfies
\[
    \EE\cSqBr{\pi_k(v)}{E_k} \geq \frac{\Paren{1 - \Paren{1 - \frac{1}{y}}^{m-1}}\Paren{1 - \frac{1}{y}}}{6\Paren{1 - \frac{1}{m}}} \cdot \EE\cSqBr{\pi_\ell(y, v_{-\ell})}{v_\ell = y, E_k}.
\]
\end{restatable}

\begin{restatable}{lemma}{ubNFprob}\label{lem:upper-bound-expected-prob-non-fav-given-val-improved}
    Let $(\pi, p)$ be a BIC mechanism for $n$ bidders with $m$ additive valuations drawn i.i.d. from $\ER$.
For any item $j$, the expected interim probability of receiving item $j$ conditioned on some item other than item $j$ being the favorite and on having $v_j = y$ satisfies
\[
    \EE\cSqBr{\pi_j(y, v_{-j})}{v_j = y, E_{-j}} \leq \frac{6m}{n\Paren{1 - \Paren{1 - \frac{1}{y}}^{m-1}}\Paren{1 - \frac{1}{y}}}.
\]
\end{restatable}


\begin{restatable}{lemma}{contributionNF}\label{lemma:contribution-NF-virtual-welfare}
Let $(\pi, p)$ be a BIC mechanism for $n$ bidders with $m$ additive valuations drawn i.i.d. from $\ER$.
The expected contribution of each non-favorite item to the virtual welfare is at most
\[
    \EE\cSqBr{v_j \pi_j(v)}{E_{-j}} \leq O\Paren{\frac{m \ln n}{n}}.
\]
\end{restatable}

Summing over all $m-1$ non-favorite items and all $n$ bidders gives an upper bound on the virtual welfare of $O(m^2\ln n)$, in addition to $nm$ from the favorite item (\cite{CaiDW16}), which establishes our new upper bound:

\begin{restatable}{theorem}{revERmn}
    $\rev_n(\ER^m) = nm + O (m^2 \ln n)$, and therefore, for some absolute constant $c$, 
    $$\vcg_{n+cm\ln(n)}(\ER^m) \geq \rev_n(\ER^m).$$
\end{restatable}

\subsection{Selling to $\ER^m$ via the Grand Bundle: $\rev_n(\ER^m) \geq nm+ \Omega(m \ln n)$}\label{sec:ER-revenue-LB}



In this section, we show that the upper bound on $\rev_n(\ER^m)$ proved in Section \ref{sec:ER-revenue-UB} is nearly tight.
More specifically, we show that selling the grand bundle via a second-price auction obtains $nm + \Theta(m \ln (mn))$ in revenue.
Note that we improve upon the previous best lower bound of $nm + \Omega(\ln n)$ from~\cite{BW18}. 

\begin{theorem}\label{theorem:ER-brev}
   $ \rev^{\text{SPA-GB}}(\ER^m) = nm + \Theta(m \ln (mn)).$
\end{theorem}




Since the revenue of a second-price auction is given by the second highest value for the item being sold, we give upper and lower bounds on the second highest value for the grand bundle.
Unfortunately, sums of random variables are difficult to work with, so we seek a good proxy for the second highest bundle value that is more straightforward to analyze.
We claim that the bundle value of the bidder with the second highest value for her favorite item is a good proxy.

To see why intuitively, first note that we expect each bidder's value for the grand bundle to be dominated by her value for her favorite item: conditioned on the value for the favorite item, we expect the value for each non-favorite item to be exponentially smaller than her value for her favorite item.

\begin{restatable}{lemma}{cENonFav}\label{lemma:conditional-expectation-non-favorite}
$\EE_{x \sim \ER}\cSqBr{x}{x \leq v} = \frac{\ln v}{1 - 1/v}$
\end{restatable}

Thus, we expect the bidder with the highest value for any item to also have the highest value for the grand bundle, the bidder with the second highest value for her favorite item to have the second highest value for the grand bundle, and so on.
In particular, we expect the bidder with the second highest value for her favorite item to set the price for the grand bundle.

Let $v_{(i),(j)}$ denote the $j$-th highest value possessed by the bidder with the $i$-th highest value for her favorite item.
Expressed in this notation, our intuition is that $\sum_j v_{(2), (j)}$ traces the second highest value for the grand bundle.
Our proof of Theorem~\ref{theorem:ER-brev} shows that this is precisely the case.
We will also see that the expectation of $v_{(2),(1)}$ is approximately $nm$.

\begin{restatable}{lemma}{ESecondHighestFav}\label{lemma:expectation-second-highest-favorite}
\[
    nm - O(m) \leq \EE[v_{(2),(1)}] \leq nm
\]
\end{restatable}

\noindent By Lemmas \ref{lemma:conditional-expectation-non-favorite} and \ref{lemma:expectation-second-highest-favorite}, we expect the second highest value for the grand bundle to be around
\[
    \EE\SqBr{\sum_j v_{(2), (j)}} = \EE\SqBr{v_{(2),(1)} + m \ln(v_{(2),(1)})} = nm + \Theta( m \ln(mn) )
\]


\subsection{``Region-separated'' Flows} \label{sec:kfbic}

Finally, we discuss the class of ``region-separated'' flows, which do not cross any axes between different favorite-item regions $R_j$. These correspond to auctions that respect all BIC constraints between bidders with the same favorite item, but not necessarily between bidders with different favorite items; we term such auctions Knows-Favorite BIC (KF-BIC). We design a KF-BIC auction that achieves revenue $nm + \Omega(m\sqrt{nm})$ from $n$ bidders with values drawn i.i.d. from $\ER^m$.

In addition to potentially being of independent interest, this result further highlights our upper bound on $\rev_n(\ER^m)$ from Section~\ref{sec:ERTsqrt} as interesting because it provably cannot follow from an ``region-separated'' flow, and establishes that the cross-diagonal BIC constraints are \emph{necessary} to achieve the optimal bound.

\begin{definition}[KF-BIC] We say that an auction is Knows-Favorite Bayesian Incentive Compatible if for all types $\vec{v}$ with distinct values for every item, $\vec{v}$ does not wish to misreport any other $\vec{w}$ with the same favorite item. That is, if $S_j$ denotes the subset of valuations in the support of $D$ such that $v_j > v_{j'}$ for all $j' \neq j$, and also $v_{j'} \neq v_{j''}$ for all $j',j''$, a mechanism with interim allocation rule $\vec{\pi}(\cdot),p(\cdot)$ is KF-BIC if:
$$\forall j,\ \forall \vec{v}, \vec{w} \in S_j,\ \vec{v} \cdot \vec{\pi}(\vec{v}) -p(\vec{v}) \geq \vec{v} \cdot \vec{\pi}(\vec{w}) - p(\vec{w}).$$
\end{definition}

Importantly, note that a KF-BIC auction is not necessarily BIC. This is both because there are no constraints that involve types with the same value for multiple items, and also because the KF-BIC constraints only guarantee that bidders do not wish to misreport \emph{while keeping their favorite item the same} (they may wish to misreport their favorite item).

\subsubsection{The Knows-Favorite Auction (KFA)} \label{sec:kfa}

\begin{definition}[Knows-Favorite Auction] The Knows-Favorite Auction (KFA) proceeds as follows: 

Let $S$ denote the set of bidders with distinct values for all $m$ items, and $S_j$ denote the subset of $S$ consisting of bidders with favorite item $j$. Each item $j$ is auctioned as follows:
\begin{itemize}
    \item If any bidder $i \in S_j$ has $v_{ij} \ge H = e^{nm}$, the item is awarded to a uniformly random such bidder, and they are charged $H$.
    
    \item If no bidder in $S_j$ has value at least $H$, then the item is offered to bidders in $S \setminus S_j$ at price $L = \sqrt{nm}$ (that is, as long as any bidder in $S \setminus S_j$ is willing to pay $L$, a uniformly random such bidder is given the item and charged $L$).
\end{itemize}
\end{definition}

\begin{restatable}{observation}{kfatruthful}\label{obs:kfa_truthful}
    KFA is KF-BIC; that is, for all $j$, no bidder in $S_j$ wishes to misreport any other type in $S_j$.
\end{restatable}

\begin{lemma} \label{lem:kfa_rev}
$\rev^{KFA}(\ER^{nm}) = nm+\Omega(m\sqrt{nm})$. That is, the expected revenue (assuming bidders tell the truth) of KFA is $nm + \Omega(m\sqrt{nm})$.
\end{lemma} 
\begin{proof}
    For each item $j$, we consider the revenue from $S_j$ and $S \setminus S_j$ separately.  
    
    Let $p$ denote the probability that a bidder $i \in S$ is in $S_j$ and has $v_{ij} \ge H$, and observe that
    \begin{align*}
        p &= \int_{x=H}^\infty \mathbb{P}(v_{ij} = x) \cdot \prod_{j' \ne j} \mathbb{P}(v_{ij'} < x) dx = \int_{x=H}^\infty \frac{1}{x^2} \left(1 - \frac1x\right)^m dx \\
        \implies p &\in \left[ \int_{x=H}^\infty \frac{1}{x^2} \left(1 - \frac mx\right) dx , \int_{x=H}^\infty \frac{1}{x^2} dx  \right] = \left[\frac1H - \frac{m}{2H^2} , \frac1H \right].
    \end{align*}
    Now, the probability that item $j$ is sold to such a bidder in $S_j$ can be bounded by the union bound:
    \begin{align*}
        \mathbb{P}(\text{sold to }S_j) &\le np \le \frac n H, \\
        \mathbb{P}(\text{sold to }S_j) &\ge np - \binom n 2 p^2 \ge \frac n H - \frac{nm}{2H^2} - \binom{n}{2} \frac{1}{H^2} = \frac n H - \frac{n(n+m-1)}{2H^2}.
    \end{align*}
    Thus, the revenue that comes from selling item $j$ to a bidder is between $n$ and $n - \frac{n(n+m-1)}{2e^{nm}}$, so the expected revenue from this case is $n - o(\frac{1}{m})$.
    
    On the other hand, the probability that item $j$ is available for bidders in $S \setminus S_j$ is at least $1 - \frac n H = 1 - o(1)$. Let $q$ denote the probability that a bidder $i \in S$ is in $S \setminus S_j$ but has $v_{ij} \ge L$, and observe that 
    \begin{align*}
        q &= \int_{x=L}^\infty \mathbb{P}(v_{ij} = x) \cdot \mathbb{P} \left(\bigcup_{j' \ne j} v_{ij'} > x \right) dx \\
        &= \int_{x=L}^\infty \frac{1}{x^2} \left(1 - \left(1-\frac{1}{x}\right)^{m-1} \right) dx \\
        \implies q&\le \int_{x=L}^\infty \frac{m-1}{x^3} dx = \frac{m-1}{2L^2}, \\
        q &\ge \int_{x=L}^\infty \left( \frac{m-1}{x^3} - \binom{m-1}{2} \frac{1}{x^4} \right) dx \\
        &= \frac{m-1}{2L^2} - \frac{(m-1)(m-2)}{6L^3}.
    \end{align*}
    Then, we have
    \begin{align*}
        \mathbb{P}(\text{sold to } S \setminus S_j) &\ge nq - \binom n 2 q^2 \\
        &\ge \frac{n(m-1)}{2L^2} - \frac{n(m-1)(m-2)}{6L^3} - \binom{n}{2} \frac{(m-1)^2}{4L^4},
    \end{align*}
    so the revenue from this case is at least
    \begin{align*}
        \frac{n(m-1)}{2L} - \frac{n(m-1)(m-2)}{6L^2} - \binom{n}{2} \frac{(m-1)^2}{4L^3} &\ge \frac{n(m-1)}{2\sqrt{nm}} - \frac{nm^2}{6nm} - \frac{n^2 m^2}{8nm\sqrt{nm}} \\
        &= \frac{\sqrt{nm}}{2} - \frac{\sqrt{n}}{2\sqrt{m}} - \frac{m}{6} - \frac{\sqrt{nm}}{8} \\
        &\ge \frac{\sqrt{nm}}{2} - \frac{\sqrt{n}}{2\sqrt{m}} - \frac{\sqrt{nm}}{6} - \frac{\sqrt{nm}}{8} \\
        &= \Omega(\sqrt{nm}).
    \end{align*}
    
    Finally, summing over all $m$ items gives a total revenue of $nm - o(1) + m \Omega(\sqrt{nm}) = nm + \Omega(m\sqrt{nm})$.
\end{proof}

\section{Conclusion}
We settle the competition complexity of $n$ bidders with additive valuations over $m$ independent items at $\Theta(\sqrt{nm})$ in the ``Big $n$'' regime. As the ``Little $n$'' regime is previously settled by~\cite{FeldmanFR18, BeyhaghiW19}, this settles the competition complexity for additive bidders over independent items (up to constant factors). On the technical front, we design an explicit BIC-but-not-DSIC mechanism outperforming selling separately (even with additional bidders) in a regime where selling separately is already a $(1-o(1))$-approximation.

We also provide results of independent interest accumulated from our journey: the competition complexity of additive bidders is exactly equal to the competition complexity when restricted to iid truncated equal revenue curves, and despite this the untruncated Equal Revenue curve witnesses an exponentially-suboptimal lower bound. 

As our work now settles the key remaining open problem for competition complexity of exceeding the optimal BIC mechanism by VCG, there are two important directions for future work:

\begin{itemize}
    \item What about the competition complexity of exceeding the optimal DSIC auction? Our BIC auctions cannot be made DSIC, and it initially seems as though BIC auctions may strictly outperform DSIC auctions for the instances that yield our main result. We suspect that our Independent Result I will be useful for upper bounds on this front (if indeed improved upper bounds are possible).
    \item What about the competition complexity of other simple auctions? There is limited work in this direction so far, which so far still loses some (small) fraction of revenue rather truly exceeding the optimum~\cite{FeldmanFR18, CaiS21}.
\end{itemize}

\bibliographystyle{alpha}
\bibliography{main.bib}

\newcommand{\etalchar}[1]{$^{#1}$}
\begin{thebibliography}{BCKW15}

\bibitem[AALS22]{AkbarpourALS22}
Mohammad Akbarpour, Yeganeh Alimohammadi, Shengwu Li, and Amin Saberi.
\newblock The value of excess supply in spatial matching markets.
\newblock In David~M. Pennock, Ilya Segal, and Sven Seuken, editors, {\em {EC} '22: The 23rd {ACM} Conference on Economics and Computation, Boulder, CO, USA, July 11 - 15, 2022}, page~62. {ACM}, 2022.

\bibitem[AMS18]{AkbarpourMS18}
Mohammad Akbarpour, Suraj Malladi, and Amin Saberi.
\newblock Diffusion, seeding, and the value of network information.
\newblock In {\'{E}}va Tardos, Edith Elkind, and Rakesh Vohra, editors, {\em Proceedings of the 2018 {ACM} Conference on Economics and Computation, Ithaca, NY, USA, June 18-22, 2018}, page 641. {ACM}, 2018.

\bibitem[BCDV22]{BrustleCDV22}
Johannes Brustle, Jos{\'{e}}~R. Correa, Paul D{\"{u}}tting, and Victor Verdugo.
\newblock The competition complexity of dynamic pricing.
\newblock In David~M. Pennock, Ilya Segal, and Sven Seuken, editors, {\em {EC} '22: The 23rd {ACM} Conference on Economics and Computation, Boulder, CO, USA, July 11 - 15, 2022}, pages 303--320. {ACM}, 2022.

\bibitem[BCKW15]{BriestCKW15}
Patrick Briest, Shuchi Chawla, Robert Kleinberg, and S.~Matthew Weinberg.
\newblock Pricing lotteries.
\newblock {\em J. Economic Theory}, 156:144--174, 2015.

\bibitem[BCU12]{BarmanCU12}
Siddharth Barman, Shuchi Chawla, and Seeun Umboh.
\newblock A bicriteria approximation for the reordering buffer problem.
\newblock In Leah Epstein and Paolo Ferragina, editors, {\em Algorithms - {ESA} 2012 - 20th Annual European Symposium, Ljubljana, Slovenia, September 10-12, 2012. Proceedings}, volume 7501 of {\em Lecture Notes in Computer Science}, pages 157--168. Springer, 2012.

\bibitem[BH11]{BeiH11}
Xiaohui Bei and Zhiyi Huang.
\newblock {Bayesian Incentive Compatibility via Fractional Assignments}.
\newblock In {\em the Twenty-Second Annual ACM-SIAM Symposium on Discrete Algorithms (SODA)}, 2011.

\bibitem[BILW20]{BabaioffILW20}
Moshe Babaioff, Nicole Immorlica, Brendan Lucier, and S.~Matthew Weinberg.
\newblock A simple and approximately optimal mechanism for an additive buyer.
\newblock {\em J. {ACM}}, 67(4):24:1--24:40, 2020.

\bibitem[BK96]{BulowK96}
Jeremy Bulow and Paul Klemperer.
\newblock Auctions versus negotiations.
\newblock {\em The American Economic Review}, pages 180--194, 1996.

\bibitem[BW18]{BW18}
Hedyeh Beyhaghi and S.~Matthew Weinberg.
\newblock Optimal (and benchmark-optimal) competition complexity for additive buyers over independent items, 2018.

\bibitem[BW19]{BeyhaghiW19}
Hedyeh Beyhaghi and S.~Matthew Weinberg.
\newblock Optimal (and benchmark-optimal) competition complexity for additive buyers over independent items.
\newblock In {\em Proceedings of the 51st ACM Symposium on Theory of Computing Conference (STOC)}, 2019.

\bibitem[CDO{\etalchar{+}}22]{ChenDOPSY22}
Xi~Chen, Ilias Diakonikolas, Anthi Orfanou, Dimitris Paparas, Xiaorui Sun, and Mihalis Yannakakis.
\newblock On the complexity of optimal lottery pricing and randomized mechanisms for a unit-demand buyer.
\newblock {\em {SIAM} J. Comput.}, 51(3):492--548, 2022.

\bibitem[CDW16]{CaiDW16}
Yang Cai, Nikhil Devanur, and S.~Matthew Weinberg.
\newblock A duality based unified approach to bayesian mechanism design.
\newblock In {\em Proceedings of the 48th ACM Conference on Theory of Computation(STOC)}, 2016.

\bibitem[CHK07]{ChawlaHK07}
Shuchi Chawla, Jason~D. Hartline, and Robert~D. Kleinberg.
\newblock {Algorithmic Pricing via Virtual Valuations}.
\newblock In {\em the 8th ACM Conference on Electronic Commerce (EC)}, 2007.

\bibitem[CHMS10]{ChawlaHMS10}
Shuchi Chawla, Jason~D. Hartline, David~L. Malec, and Balasubramanian Sivan.
\newblock {Multi-Parameter Mechanism Design and Sequential Posted Pricing}.
\newblock In {\em the 42nd ACM Symposium on Theory of Computing (STOC)}, 2010.

\bibitem[CHMS13]{ChawlaHMS13}
Shuchi Chawla, Jason Hartline, David Malec, and Balasubramanian Sivan.
\newblock {Prior-Independent Mechanisms for Scheduling}.
\newblock In {\em Proceedings of 45th ACM Symposium on Theory of Computing (STOC)}, 2013.

\bibitem[Cla71]{Clarke71}
Edward~H. Clarke.
\newblock {Multipart Pricing of Public Goods}.
\newblock {\em Public Choice}, 11(1):17--33, 1971.

\bibitem[CM16]{ChawlaM16}
Shuchi Chawla and J.~Benjamin Miller.
\newblock Mechanism design for subadditive agents via an ex ante relaxation.
\newblock In {\em Proceedings of the 2016 {ACM} Conference on Economics and Computation, {EC} '16, Maastricht, The Netherlands, July 24-28, 2016}, pages 579--596, 2016.

\bibitem[CMS15]{ChawlaMS15}
Shuchi Chawla, David~L. Malec, and Balasubramanian Sivan.
\newblock The power of randomness in bayesian optimal mechanism design.
\newblock {\em Games and Economic Behavior}, 91:297--317, 2015.

\bibitem[COVZ21]{CaiOVZ21}
Yang Cai, Argyris Oikonomou, Grigoris Velegkas, and Mingfei Zhao.
\newblock An efficient {\(\epsilon\)}-bic to {BIC} transformation and its application to black-box reduction in revenue maximization.
\newblock In D{\'{a}}niel Marx, editor, {\em Proceedings of the 2021 {ACM-SIAM} Symposium on Discrete Algorithms, {SODA} 2021, Virtual Conference, January 10 - 13, 2021}, pages 1337--1356. {SIAM}, 2021.

\bibitem[COZ22]{CaiOZ22}
Yang Cai, Argyris Oikonomou, and Mingfei Zhao.
\newblock Computing simple mechanisms: Lift-and-round over marginal reduced forms.
\newblock In Stefano Leonardi and Anupam Gupta, editors, {\em {STOC} '22: 54th Annual {ACM} {SIGACT} Symposium on Theory of Computing, Rome, Italy, June 20 - 24, 2022}, pages 704--717. {ACM}, 2022.

\bibitem[CS21]{CaiS21}
Linda Cai and Raghuvansh~R. Saxena.
\newblock 99{\%} revenue with constant enhanced competition.
\newblock In P{\'{e}}ter Bir{\'{o}}, Shuchi Chawla, and Federico Echenique, editors, {\em {EC} '21: The 22nd {ACM} Conference on Economics and Computation, Budapest, Hungary, July 18-23, 2021}, pages 224--241. {ACM}, 2021.

\bibitem[CZ17]{CaiZ17}
Yang Cai and Mingfei Zhao.
\newblock Simple mechanisms for subadditive buyers via duality.
\newblock In {\em Proceedings of the 49th Annual {ACM} {SIGACT} Symposium on Theory of Computing, {STOC} 2017, Montreal, QC, Canada, June 19-23, 2017}, pages 170--183, 2017.

\bibitem[DDT14]{DaskalakisDT14}
Constantinos Daskalakis, Alan Deckelbaum, and Christos Tzamos.
\newblock {The Complexity of Optimal Mechanism Design}.
\newblock In {\em the 25th ACM-SIAM Symposium on Discrete Algorithms (SODA)}, 2014.

\bibitem[DDT17]{DaskalakisDT17}
Constantinos Daskalakis, Alan Deckelbaum, and Christos Tzamos.
\newblock Strong duality for a multiple-good monopolist.
\newblock {\em Econometrica}, 85(3):735--767, 2017.

\bibitem[DHKN17]{DughmiHKN17}
Shaddin Dughmi, Jason~D. Hartline, Robert Kleinberg, and Rad Niazadeh.
\newblock Bernoulli factories and black-box reductions in mechanism design.
\newblock In {\em Proceedings of the 49th Annual {ACM} {SIGACT} Symposium on Theory of Computing, {STOC} 2017, Montreal, QC, Canada, June 19-23, 2017}, pages 158--169, 2017.

\bibitem[DW12]{DaskalakisW12}
Constantinos Daskalakis and S.~Matthew Weinberg.
\newblock Symmetries and optimal multi-dimensional mechanism design.
\newblock In {\em Proceedings of the 13th {ACM} Conference on Electronic Commerce, {EC} 2012, Valencia, Spain, June 4-8, 2012}, pages 370--387, 2012.

\bibitem[EFF{\etalchar{+}}17]{EdenFFTW17b}
Alon Eden, Michal Feldman, Ophir Friedler, Inbal Talgam{-}Cohen, and S.~Matthew Weinberg.
\newblock The competition complexity of auctions: {A} bulow-klemperer result for multi-dimensional bidders.
\newblock In {\em Proceedings of the 2017 {ACM} Conference on Economics and Computation, {EC} '17, Cambridge, MA, USA, June 26-30, 2017}, page 343, 2017.

\bibitem[EFF{\etalchar{+}}21]{EdenFFTW21}
Alon Eden, Michal Feldman, Ophir Friedler, Inbal Talgam{-}Cohen, and S.~Matthew Weinberg.
\newblock A simple and approximately optimal mechanism for a buyer with complements.
\newblock {\em Oper. Res.}, 69(1):188--206, 2021.

\bibitem[FFR18]{FeldmanFR18}
Michal Feldman, Ophir Friedler, and Aviad Rubinstein.
\newblock 99{\%} revenue via enhanced competition.
\newblock In {\em Proceedings of the 2018 {ACM} Conference on Economics and Computation, Ithaca, NY, USA, June 18-22, 2018}, pages 443--460, 2018.

\bibitem[FLR19]{FuLR19}
Hu~Fu, Christopher Liaw, and Sikander Randhawa.
\newblock The vickrey auction with a single duplicate bidder approximates the optimal revenue.
\newblock In Anna Karlin, Nicole Immorlica, and Ramesh Johari, editors, {\em Proceedings of the 2019 {ACM} Conference on Economics and Computation, {EC} 2019, Phoenix, AZ, USA, June 24-28, 2019}, pages 419--420. {ACM}, 2019.

\bibitem[GK14]{GiannakopoulosK14}
Yiannis Giannakopoulos and Elias Koutsoupias.
\newblock Duality and optimality of auctions for uniform distributions.
\newblock In {\em {ACM} Conference on Economics and Computation, {EC} '14, Stanford , CA, USA, June 8-12, 2014}, pages 259--276, 2014.

\bibitem[GK15]{GiannakopoulosK15}
Yiannis Giannakopoulos and Elias Koutsoupias.
\newblock Selling two goods optimally.
\newblock In {\em Automata, Languages, and Programming - 42nd International Colloquium, {ICALP} 2015, Kyoto, Japan, July 6-10, 2015, Proceedings, Part {II}}, pages 650--662, 2015.

\bibitem[Gro73]{Groves73}
Theodore Groves.
\newblock {Incentives in Teams}.
\newblock {\em Econometrica}, 41(4):617--631, 1973.

\bibitem[HH15]{HartlineH15}
Nima Haghpanah and Jason Hartline.
\newblock Reverse mechanism design.
\newblock In {\em Proceedings of the Sixteenth {ACM} Conference on Economics and Computation, {EC} '15, Portland, OR, USA, June 15-19, 2015}, 2015.

\bibitem[HKM11]{HartlineKM11}
Jason~D. Hartline, Robert Kleinberg, and Azarakhsh Malekian.
\newblock {Bayesian Incentive Compatibility via Matchings}.
\newblock In {\em the Twenty-Second Annual ACM-SIAM Symposium on Discrete Algorithms (SODA)}, 2011.

\bibitem[HN13]{HartN13}
Sergiu Hart and Noam Nisan.
\newblock The menu-size complexity of auctions.
\newblock In {\em the 14th ACM Conference on Electronic Commerce (EC)}, 2013.

\bibitem[HN17]{HartN17}
Sergiu Hart and Noam Nisan.
\newblock Approximate revenue maximization with multiple items.
\newblock {\em J. Economic Theory}, 172:313--347, 2017.

\bibitem[HR09]{HartlineR09}
Jason~D. Hartline and Tim Roughgarden.
\newblock Simple versus optimal mechanisms.
\newblock In {\em ACM Conference on Electronic Commerce}, pages 225--234, 2009.

\bibitem[HR15]{HartR15}
Sergiu Hart and Philip~J. Reny.
\newblock {Maximizing Revenue with Multiple Goods: Nonmonotonicity and Other Observations}.
\newblock {\em Theoretical Economics}, 10(3):893--922, 2015.

\bibitem[LP18]{LiuP18}
Siqi Liu and Christos{-}Alexandros Psomas.
\newblock On the competition complexity of dynamic mechanism design.
\newblock In {\em Proceedings of the Twenty-Ninth Annual {ACM-SIAM} Symposium on Discrete Algorithms, {SODA} 2018, New Orleans, LA, USA, January 7-10, 2018}, pages 2008--2025, 2018.

\bibitem[LY13]{LiY13}
Xinye Li and Andrew Chi-Chih Yao.
\newblock On revenue maximization for selling multiple independently distributed items.
\newblock {\em Proceedings of the National Academy of Sciences}, 110(28):11232--11237, 2013.

\bibitem[Mye81]{Myerson81}
Roger~B. Myerson.
\newblock {Optimal Auction Design}.
\newblock {\em Mathematics of Operations Research}, 6(1):58--73, 1981.

\bibitem[Pav11]{Pavlov11}
Gregory Pavlov.
\newblock Optimal mechanism for selling two goods.
\newblock {\em The B.E. Journal of Theoretical Economics}, 11(3), 2011.

\bibitem[PSW19]{PsomasSW19}
Alexandros Psomas, Ariel Schvartzman, and S.~Matthew Weinberg.
\newblock Smoothed analysis of multi-item auctions with correlated values.
\newblock In {\em Proceedings of the 2019 {ACM} Conference on Economics and Computation, {EC} 2019, Phoenix, AZ, USA, June 24-28, 2019.}, pages 417--418, 2019.

\bibitem[PSW22]{PsomasSW22}
Alexandros Psomas, Ariel Schvartzman, and S.~Matthew Weinberg.
\newblock On infinite separations between simple and optimal mechanisms.
\newblock In {\em NeurIPS}, 2022.

\bibitem[RC98]{RochetC98}
Jean-Charles Rochet and Philippe Chone.
\newblock Ironing, sweeping, and multidimensional screening.
\newblock {\em Econometrica}, 66(4):783--826, 1998.

\bibitem[RT02]{RoughgardenT02}
Tim Roughgarden and {\'{E}}va Tardos.
\newblock How bad is selfish routing?
\newblock {\em J. {ACM}}, 49(2):236--259, 2002.

\bibitem[RTCY12]{RoughgardenTY12}
Tim Roughgarden, Inbal Talgam-Cohen, and Qiqi Yan.
\newblock Supply-limiting mechanisms.
\newblock In {\em 13th ACM Conference on Electronic Commerce (EC)}, 2012.

\bibitem[RW15]{RubinsteinW15}
Aviad Rubinstein and S.~Matthew Weinberg.
\newblock Simple mechanisms for a subadditive buyer and applications to revenue monotonicity.
\newblock In {\em Proceedings of the Sixteenth {ACM} Conference on Economics and Computation, {EC} '15, Portland, OR, USA, June 15-19, 2015}, pages 377--394, 2015.

\bibitem[ST85]{SleatorT85}
Daniel~Dominic Sleator and Robert~Endre Tarjan.
\newblock Amortized efficiency of list update and paging rules.
\newblock {\em Commun. {ACM}}, 28(2):202--208, 1985.

\bibitem[Tha04]{Thanassoulis04}
John Thanassoulis.
\newblock Haggling over substitutes.
\newblock {\em Journal of Economic Theory}, 117:217--245, 2004.

\bibitem[Vic61]{Vickrey61}
William Vickrey.
\newblock {Counterspeculations, Auctions, and Competitive Sealed Tenders}.
\newblock {\em Journal of Finance}, 16(1):8--37, 1961.

\bibitem[WZ22]{WeinbergZ22}
S.~Matthew Weinberg and Zixin Zhou.
\newblock Optimal multi-dimensional mechanisms are not locally-implementable.
\newblock In David~M. Pennock, Ilya Segal, and Sven Seuken, editors, {\em {EC} '22: The 23rd {ACM} Conference on Economics and Computation, Boulder, CO, USA, July 11 - 15, 2022}, pages 875--896. {ACM}, 2022.

\bibitem[Yao15]{Yao15}
Andrew Chi-Chih Yao.
\newblock {An n-to-1 bidder reduction for multi-item auctions and its applications}.
\newblock In {\em the Twenty-Sixth Annual ACM-SIAM Symposium on Discrete Algorithms (SODA)}, 2015.

\bibitem[Yao17]{Yao17}
Andrew~Chi{-}Chih Yao.
\newblock Dominant-strategy versus bayesian multi-item auctions: Maximum revenue determination and comparison.
\newblock In Constantinos Daskalakis, Moshe Babaioff, and Herv{\'{e}} Moulin, editors, {\em Proceedings of the 2017 {ACM} Conference on Economics and Computation, {EC} '17, Cambridge, MA, USA, June 26-30, 2017}, pages 3--20. {ACM}, 2017.

\bibitem[Yao18]{Yao18}
Andrew~Chi{-}Chih Yao.
\newblock An incentive analysis of some bitcoin fee designs.
\newblock {\em CoRR}, abs/1811.02351, 2018.

\end{thebibliography}

\appendix
\section{Deferred proofs}\label{sec:proofs}
Here, we provide complete proofs deferred from the main text.

\subsection{Proofs from Section~\ref{sec:ERTsqrt}} \label{sec:proofs_ERTsqrt}

\srevcompute*
\begin{proof}
    Observe that $\ER_{<T}$ is a regular distribution, with 
    \begin{align*}
        \bar{\varphi}_{j} (v_{ij}) &= v_{ij} - \frac{1 - (1 - 1/v_{ij})}{1/v_{ij}^2} = 0 \qquad \forall v_{ij} < T, \\
        \bar{\varphi}_{j} (T) &= T - 0 = T,
    \end{align*}
    so the optimal auction for each item allocates the item to a bidder with value $T$ for price $T$.
\end{proof}

\begin{restatable}{corollary}{bidderWT}\label{corollary:bidder-w-T}
Suppose $T \geq \sqrt{mn}$ and $a \geq b$.
Let $v \in [1,T]^m$ and suppose there exists some item with value $T$.
Define $H := \CrBr{j : v_j = T}$ and $L := \CrBr{j : v_j \geq mn/T} \setminus H$.
A bidder with type $v$ prefers the menu option $(H, L)$ over any other option in the menu (including not receiving any items).
\end{restatable}

\begin{proof}
By Lemma \ref{lemma:preferred-option}, a bidder with such a type prefers $(H,L)$ over any other option that allocates an item with some positive probability.
It remains to show that the utility of such a bidder for $(H,L)$ is non-negative:
\begin{align*}
    a \sum_{j \in H} v_j + b \sum_{j \in L} v_j & - \Paren{\abs{H}aT + \abs{L} b \frac{mn}{T} - (\abs{H} - 1) b \Paren{T - \frac{mn}{T}}} \\
        &= b \sum_{j \in L} \Paren{v_j - \frac{mn}{T}} + (\abs{H} - 1) b \Paren{T - \frac{mn}{T}} \tag{$v_j = T$ for all $j \in H$} \\
        &\geq (\abs{H} - 1) b \Paren{T - \frac{mn}{T}} \tag{$v_j \geq mn/T$ for all $j \in L$} \\
        &\geq 0. \tag{$T^2 \geq mn$}
\end{align*}
\end{proof}

\begin{restatable}{corollary}{bidderWoT}\label{corollary:bidder-wo-T}
Suppose $T \geq \sqrt{mn}$ and $a \geq b$.
Let $v \in [1,T]^m$.
Let $j^* \in \arg\max_j v_j$ and suppose $v_{j^*} < T$.
Define $L := \CrBr{j : v_j \geq mn/T} \setminus \{j^*\}$.
A bidder with type $v$ prefers the menu option $(\{j^*\}, L)$ over any other option in the menu (including not receiving any items) if and only if $av_{j^*} + b \sum_{j \in L} v_j \geq aT + \abs{L}b\frac{mn}{T}$.
\end{restatable}

\begin{proof}
By Lemma \ref{lemma:preferred-option}, a bidder with such a type prefers $(\{j^*\},L)$ over any other option that allocates an item with some positive probability.
To conclude, note that the utility for $(\{j^*\},L)$ is non-negative if and only if the inequality in the lemma statement holds.
\end{proof}

\aVsB*
\begin{proof}
We have
\begin{align*}
    \frac{b}{a} 
        &= \frac{\PP_{v_{-i}}\SqBr{\text{$\not\exists$ high}} \Paren{1 - \Paren{1 - \PP\cSqBr{\text{low}}{\text{not high}}}^n}}{n\PP\cSqBr{\text{low}}{\text{not high}}} \cdot \frac{n\PP\SqBr{\text{high}}}{1 - \Paren{1 - \PP\SqBr{\text{high}}}^n} \\
        &\leq \frac{\PP_{v_{-i}}\SqBr{\text{$\not\exists$ high}} \cdot n \PP \cSqBr{\text{low}}{\text{not high}}} {n\PP\cSqBr{\text{low}}{\text{not high}}} \cdot \frac{n\PP\SqBr{\text{high}}}{1 - \Paren{1 - \PP\SqBr{\text{high}}}^n} \tag{union bound} \\
        &= \frac{n \PP_{v_{-i}}\SqBr{\text{$\not\exists$ high}} \PP\SqBr{\text{high}}}{1 - \Paren{1 - \PP\SqBr{\text{high}}}^n} \\
        &= \frac{n \Paren{1 - \frac{1}{T} - \sum\limits_{\ell = 1}^{m-1} \binom{m - 1}{\ell} q_\ell}^{n-1} \Paren{\frac{1}{T} + \sum\limits_{\ell = 1}^{m-1} \binom{m - 1}{\ell} q_\ell}}{1 - \Paren{1 - \frac{1}{T} - \sum\limits_{\ell = 1}^{m-1} \binom{m - 1}{\ell} q_\ell}^n} \\
        &\leq \frac{n \Paren{1 - \frac{1}{T}}^{n}}{T \Paren{1 - \Paren{1 - \frac{1}{T}}^n}}  \tag{replace summation with $0$} \\
        &\leq \frac{\frac{n}{T} e^{-\frac{n}{T}}}{1 - e^{-\frac{n}{T}}}.
\end{align*}
The second inequality follows from the fact that $\Paren{1 - \frac{1}{T} - x}^{n-1} \Paren{\frac{1}{T} + x}/\Paren{1 - \Paren{1 - \frac{1}{T} - x}^n}$ is decreasing in $x$ for $x \in [0, 1 - 1/T]$ (note that $1 - \frac{1}{T} - \sum_{\ell = 1}^{m-1} \binom{m - 1}{\ell} q_\ell \le 1$ because the left-hand expression is a probability, so indeed $\sum_{\ell = 1}^{m-1} \binom{m - 1}{\ell} q_\ell \leq 1 - \frac{1}{T}$):
\begin{align*}
    \frac{\partial}{\partial x} \frac{\Paren{1 - \frac{1}{T} - x}^{n-1} \Paren{\frac{1}{T} + x}}{1 - \Paren{1 - \frac{1}{T} - x}^n} 
        = \frac{\Paren{1 - \frac{1}{T} - x}^{n-2} \Paren{\Paren{1 - \Paren{1 - \frac{1}{T} - x}^{n}} - n\Paren{\frac{1}{T} + x}}}{\Paren{1 - \Paren{1 - \frac{1}{T} - x}^{n}}^2} 
        \leq 0. \tag{union bound; $x \in [0, 1-1/T]$}
\end{align*}
\end{proof}

\NSNvsLN*
\begin{proof}
By Corollaries \ref{corollary:bidder-w-T} and \ref{corollary:bidder-wo-T}, the revenue \textit{per bidder} of the Not-So-Naive Auction  is
\begin{align*}
    & \begin{aligned}[t]
        & \textstyle \sum\limits_{j^* \in [m]} \sum\limits_{L \subseteq [m] \setminus \{j^*\}} \Paren{aT + \abs{L}b\frac{mn}{T}} \Paren{\frac{1}{T} \Paren{\frac{T}{mn} - \frac{1}{T}}^{\abs{L}}\Paren{1 - \frac{T}{mn}}^{m - 1 - \abs{L}} + q_{\abs{L}}} \\
        & + \textstyle \sum\limits_{\substack{H \subseteq [m] : \\ \abs{H} \geq 2}} \sum\limits_{L \subseteq [m] \setminus H} \Paren{\abs{H}aT + \abs{L}b\frac{mn}{T} - (\abs{H} - 1)b\Paren{T - \frac{mn}{T}}} \frac{1}{T^{\abs{H}}}\Paren{\frac{T}{mn} - \frac{1}{T}}^{\abs{L}}\Paren{1 - \frac{T}{mn}}^{m - \abs{H} - \abs{L}} 
    \end{aligned} \\
        &= \begin{aligned}[t]
            & \textstyle \sum\limits_{k=1}^m \sum\limits_{\ell = 0}^{m - k} \binom{m}{k} \binom{m - k}{\ell} \Paren{kaT + \ell b \frac{mn}{T} - (k-1)b\Paren{T - \frac{mn}{T}}} \frac{1}{T^k}\Paren{\frac{T}{mn} - \frac{1}{T}}^\ell \Paren{1 - \frac{T}{mn}}^{m - k - \ell} \\
            & + \textstyle m \sum\limits_{\ell = 1}^{m-1} \binom{m - 1}{\ell}  \Paren{aT + \ell b \frac{mn}{T}} q_\ell
        \end{aligned} \\
        &= \begin{aligned}[t]
            & \textstyle am + bmn \Paren{\frac{T}{mn} - \frac{1}{T}} \Paren{1 - \Paren{1 - \frac{1}{T}}^m - \frac{m}{T}\Paren{1 - \frac{1}{T}}^{m-1}} + m \sum\limits_{\ell = 1}^{m-1} \binom{m - 1}{\ell} \Paren{aT + \ell b \frac{mn}{T}} q_\ell.
        \end{aligned}
\end{align*}
Recall that the revenue extracted by the Less-Naive Auction from selling item $j$ to bidder $i$ is
\begin{align*}
    a_0 T \cdot \1\Paren{v_{ij} = T, \max_{j' < j} v_{ij'} < T} & + a_0 \Paren{T - \frac{b_0}{a_0}\Paren{T - \frac{mn}{T}}} \cdot \1\Paren{v_{ij} = T, \max_{j' < j} v_{ij'} = T} \\
        &+ \frac{b_0 mn}{T} \cdot \1\Paren{v_{ij} \in [\tfrac{mn}{T}, T), \max_{j'\ne j} v_{ij'} = T}.
\end{align*}

Taking the expectation over the randomness of bidder $i$'s type and summing over all items yields that the revenue extracted by the Less-Naive Auction \textit{per bidder} is
\begin{align*}
    a_0 m + b_0mn \Paren{\frac{T}{mn} - \frac{1}{T}} \Paren{1 - \Paren{1 - \frac{1}{T}}^m - \frac{m}{T}\Paren{1 - \frac{1}{T}}^{m-1}}.
\end{align*}

Now, recall that $a_0$ is the interim probability of winning an item when bidding $T$ under the allocation rule that allocates items uniformly at random to bidders with value $T$, so
\[
    a_0 = \underset{v_{-i}}{\EE}\SqBr{\frac{1}{1 + \sum_{k \not= i} \1(v_{k j} = T)}} = \sum_{k = 0}^{n-1} \frac{\binom{n-1}{k}}{k+1} \frac{1}{T^k}\Paren{1 - \frac{1}{T}}^{m-k-1} = \frac{T}{n}\Paren{1 - \Paren{1 - \frac{1}{T}}^n}.
\]
Note how $a_0$ compares to $a$.
\begin{align*}
    a\Paren{1 +  T \sum\limits_{\ell = 1}^{m-1} \binom{m - 1}{\ell} q_\ell} 
        &= \frac{\Paren{1 - \Paren{1 - \frac{1}{T} - \sum_{\ell = 1}^{m-1} \binom{m-1}{\ell} q_\ell}^n}}{n \Paren{\frac{1}{T} + \sum_{\ell = 1}^{m-1} \binom{m-1}{\ell} q_\ell}}\Paren{1 +  T \sum\limits_{\ell = 1}^{m-1} \binom{m - 1}{\ell} q_\ell} \tag{definition of $a$} \\
        &= \frac{T}{n}\Paren{1 - \Paren{1 - \frac{1}{T} - \sum\limits_{\ell = 1}^{m-1} \binom{m-1}{\ell} q_\ell}^n} \\
        &\geq \frac{T}{n}\Paren{1 - \Paren{1 - \frac{1}{T}}^n} \\
        &= a_0 \tag{definition of $a_0$}
\end{align*}
Thus, the revenue \textit{per bidder} of the Not-So-Naive Auction exceeds that of the Less-Naive Auction by
\begin{align*}
    & \textstyle (a - a_0) m + (b-b_0)mn \Paren{\frac{T}{mn} - \frac{1}{T}} \Paren{1 - \Paren{1 - \frac{1}{T}}^m - \frac{m}{T}\Paren{1 - \frac{1}{T}}^{m-1}} + m \sum\limits_{\ell = 1}^{m-1} \binom{m - 1}{\ell}  \Paren{aT + \ell b \frac{mn}{T}} q_\ell \\
        &= \begin{aligned}[t]
            & \textstyle \Paren{a\Paren{1 +  T \sum\limits_{\ell = 1}^{m-1} \binom{m - 1}{\ell} q_\ell} - a_0} m + (b-b_0)mn \Paren{\frac{T}{mn} - \frac{1}{T}} \Paren{1 - \Paren{1 - \frac{1}{T}}^m - \frac{m}{T}\Paren{1 - \frac{1}{T}}^{m-1}} \\
            & + \textstyle \frac{bm^2(m-1)n}{T} \sum\limits_{\ell = 1}^{m-1} \binom{m - 2}{\ell-1} q_\ell
        \end{aligned} \\
        &\geq \textstyle (b-b_0)mn \Paren{\frac{T}{mn} - \frac{1}{T}} \Paren{1 - \Paren{1 - \frac{1}{T}}^m - \frac{m}{T}\Paren{1 - \frac{1}{T}}^{m-1}} + \frac{bm^2(m-1)n}{T} \sum\limits_{\ell = 1}^{m-1} \binom{m - 2}{\ell-1} q_\ell. 
\end{align*}
Summing over all bidders yields the lemma.
\end{proof}

\begin{restatable}{lemma}{qellUB}\label{lemma:q_ell-upper-bound}
If $T \geq \sqrt{mn}$, then for all $\ell \in [m-1]$,
\[
    q_\ell \leq \Paren{\frac{T}{mn} - \frac{1}{T}}^{\ell} \Paren{1 - \frac{T}{mn}}^{m - \ell - 1} \frac{\ell b}{Ta}. 
\]
\end{restatable}

\begin{proof}
Observe that
\begin{align*}
    q_\ell 
        &= \textstyle \underset{v}{\PP}\SqBr{v_1 = \max\limits_j v_j < T, \min\limits_{j \in [\ell + 1]} v_j \geq \frac{mn}{T}, \max\limits_{j \not\in [\ell + 1]} v_j < \frac{mn}{T}, av_{1} + b \sum\limits_{j = 2}^{\ell + 1} v_j \geq aT + \ell b\frac{mn}{T}} \\
        & = \textstyle \frac{\Paren{\frac{T}{mn} - \frac{1}{T}}^{\ell + 1} \Paren{1 - \frac{T}{mn}}^{m - \ell - 1}}{\ell + 1} \textstyle \PP\cSqBr{a v_1 + b \sum\limits_{j = 2}^{\ell + 1} v_j \geq aT + \ell b \frac{mn}{T}}{v_1 = \max\limits_j v_j < T, \min\limits_{j \in [\ell + 1]} v_j \geq \frac{mn}{T} > \max\limits_{j \not\in [\ell + 1]} v_j} \\
        & \leq \textstyle \frac{\Paren{\frac{T}{mn} - \frac{1}{T}}^{\ell + 1} \Paren{1 - \frac{T}{mn}}^{m - \ell - 1}}{\ell + 1}  \PP\cSqBr{v_1 \geq \frac{aT + \ell b \frac{mn}{T}}{a + \ell b}}{v_1 = \max\limits_j v_j < T, \min\limits_{j \in [\ell + 1]} v_j \geq \frac{mn}{T} > \max\limits_{j \not\in [\ell + 1]} v_j} \\
        & = \textstyle \frac{\Paren{\frac{T}{mn} - \frac{1}{T}}^{\ell + 1} \Paren{1 - \frac{T}{mn}}^{m - \ell - 1}}{\ell + 1} \Paren{1 - \Paren{1 - \frac{\frac{1}{T - \frac{\ell b}{a + \ell b} \Paren{T - \frac{mn}{T}}} - \frac{1}{T}}{\frac{T}{mn} - \frac{1}{T}}}^{\ell + 1}} \\
        &\leq \textstyle  \Paren{\frac{T}{mn} - \frac{1}{T}}^{\ell} \Paren{1 - \frac{T}{mn}}^{m - \ell - 1} \Paren{\frac{1}{T - \frac{\ell b}{a + \ell b} \Paren{T - \frac{mn}{T}}} - \frac{1}{T}} \tag{union bound} \\
        &\leq \textstyle \Paren{\frac{T}{mn} - \frac{1}{T}}^{\ell} \Paren{1 - \frac{T}{mn}}^{m - \ell - 1} \frac{\ell b}{Ta}. 
\end{align*}
\end{proof}

\begin{restatable}{lemma}{probNoTHigh}\label{lemma:probability-no-T-high}
If $T = \lambda \sqrt{mn}$ for some constant $\lambda > 1$, then the probability of a ``high'' bidder with all values below $T$ can be upper bounded as
\[
    \sum_{\ell=1}^{m-1} \binom{m-1}{\ell} q_\ell \leq \Paren{1 - \frac{1}{\lambda^2}} \frac{b(m-1)}{amn} \Paren{1 - \frac{1}{T}}^{m-2}.
\]
\end{restatable}

\begin{proof}
\begin{align*}
    \sum_{\ell=1}^{m-1} \binom{m-1}{\ell} q_\ell
        &\leq \frac{b(m-1)}{Ta} \sum_{\ell=1}^{m-1}\binom{m-2}{\ell - 1} \Paren{\frac{T}{mn} - \frac{1}{T}}^\ell \Paren{1 - \frac{T}{mn}}^{m - \ell - 1} \tag{Lemma \ref{lemma:q_ell-upper-bound}} \\
        &= \frac{b(m-1)}{Ta}\Paren{\frac{T}{mn} - \frac{1}{T}} \sum_{\ell = 0}^{m-2} \binom{m-2}{\ell}\Paren{\frac{T}{mn}  - \frac{1}{T}}^\ell \Paren{1 - \frac{T}{mn}}^{m - \ell - 2} \\
        &= \frac{b(m-1)}{Ta} \Paren{\frac{T}{mn} - \frac{1}{T}}\Paren{1 - \frac{1}{T}}^{m-2} \\
        &= \Paren{1 - \frac{1}{\lambda^2}} \frac{b(m-1)}{amn} \Paren{1 - \frac{1}{T}}^{m-2}.
\end{align*}
\end{proof}

\begin{restatable}{lemma}{probNoTLow}\label{lemma:probability-no-T-low}
If $T = \lambda \sqrt{mn}$ for some constant $\lambda > 1$, then the probability of a ``low'' bidder with all values below $T$ can be upper bounded as
\[
    (m-1) \sum_{\ell=1}^{m-1} \binom{m-2}{\ell - 1} q_\ell \leq \Paren{1 - \frac{1}{\lambda^2}} \frac{b(m-1)}{amn} \Paren{1 - \frac{1}{T}}^{m-3} \Paren{1 - \frac{1}{\lambda \sqrt{mn}} + \Paren{\lambda - \frac{1}{\lambda}}\frac{m-2}{\sqrt{mn}} }.
\]    
\end{restatable}

\begin{proof}
\begin{align*}
    (m-1) \sum_{\ell=1}^{m-1} \binom{m-2}{\ell - 1} q_\ell 
        \leq {} & \frac{b(m-1)}{Ta} \sum_{\ell=1}^{m-1} \ell \binom{m-2}{\ell - 1} \Paren{\frac{T}{mn} - \frac{1}{T}}^\ell \Paren{1 - \frac{T}{mn}}^{m - \ell - 1} \\
        = {} & \frac{b(m-1)}{Ta} \Paren{\frac{T}{mn} - \frac{1}{T}}\Paren{1 - \frac{1}{T}}^{m-2} \\
            & + \frac{b(m-1)(m-2)}{Ta} \sum_{\ell = 2}^{m-1} \binom{m-3}{\ell - 2} \Paren{\frac{T}{mn} - \frac{1}{T}}^\ell \Paren{1 - \frac{T}{mn}}^{m - \ell - 1} \\
        = {} & \frac{b(m-1)}{Ta}\Paren{\frac{T}{mn} - \frac{1}{T}} \Paren{1 - \frac{1}{T}}^{m-3} \Paren{1 - \frac{1}{T} + (m-2)\Paren{\frac{T}{mn} - \frac{1}{T}}} \\
        = {} & \Paren{1 - \frac{1}{\lambda^2}} \frac{b(m-1)}{amn} \Paren{1 - \frac{1}{T}}^{m-3} \Paren{1 - \frac{1}{\lambda \sqrt{mn}} + \Paren{\lambda - \frac{1}{\lambda}}\frac{m-2}{\sqrt{mn}} }.
\end{align*}
\end{proof}

\bLB*
\begin{proof}
Note that
\[
    b := \frac{\PP_{v_{-i}}\SqBr{\text{$\not\exists$ high}} \Paren{1 - \Paren{1 - \PP\cSqBr{\text{low}}{\text{not high}}}^n}}{n\PP\cSqBr{\text{low}}{\text{not high}}} = \frac{\PP_{v}\SqBr{\text{$\not\exists$ high}} \Paren{1 - \Paren{1 - \PP\cSqBr{\text{low}}{\text{not high}}}^n}}{n\PP\SqBr{\text{low}}}
\]
We bound each part of $b$:
\begin{align*}
    \underset{v}{\PP}\SqBr{\text{$\not\exists$ high}} 
        &= \Paren{1 - \frac{1}{T} - \sum_{\ell = 1}^{m-1} \binom{m-1}{\ell} q_\ell}^n \\
        &\geq \Paren{1 - \frac{1}{T} - \frac{O(1)}{n}}^n \tag{Lemmas \ref{lemma:a-geq-b} and \ref{lemma:probability-no-T-high}} \\
        &= \Omega\Paren{\Paren{1 - \frac{1}{T}}^n} \tag{$T = \lambda\sqrt{mn}$, $m \leq n$},
\end{align*}
\begin{align*}
    \PP\SqBr{\text{low}}
        &= \Paren{1 - \Paren{1 - \frac{1}{T}}^{m - 1}} \Paren{\frac{T}{mn} - \frac{1}{T}} + (m-1) \sum_{\ell = 1}^{m-1} \binom{m-2}{\ell - 1} q_{\ell} \\
        &\leq \frac{m-1}{T}\Paren{\frac{T}{mn} - \frac{1}{T}} + \frac{O(1)}{n} \tag{union bound; Lemmas \ref{lemma:a-geq-b} and \ref{lemma:probability-no-T-low}} \\
        &= \Paren{1 - \frac{1}{\lambda^2}} \frac{m-1}{mn} + \frac{O(1)}{n}  \tag{$T = \lambda \sqrt{mn}$} \\
        &= O\Paren{\frac{1}{n}},
\end{align*}
\begin{align*}
    \PP\cSqBr{\text{low}}{\text{not high}}
        &= \frac{\Paren{1 - \Paren{1 - \frac{1}{T}}^{m - 1}} \Paren{\frac{T}{mn} - \frac{1}{T}} + (m-1) \sum_{\ell = 0}^{m-2} \binom{m-2}{\ell} q_{\ell + 1}}{1 - \frac{1}{T} - \sum_{\ell = 1}^{m-1} \binom{m - 1}{\ell} q_\ell} \\
        &\geq \frac{\Paren{1 - \Paren{1 - \frac{1}{T}}^{m - 1}} \Paren{\frac{T}{mn} - \frac{1}{T}}}{1 - \frac{1}{T}} \\
        &\geq \frac{\Paren{\frac{m-1}{T} - \frac{\binom{m-1}{2}}{T^2}} \Paren{\frac{T}{mn} - \frac{1}{T}}}{1 - \frac{1}{T}} \tag{inclusion-exclusion} \\
        &= \frac{\Paren{1 - \frac{1}{\lambda^2}} \frac{m-1}{\sqrt{mn}} \Paren{\frac{1}{\sqrt{mn}} - \frac{1}{2\lambda n}}}{1 - \frac{1}{\lambda \sqrt{mn}}} \tag{$T = \lambda \sqrt{mn}$} \\
        &= \Omega\Paren{\frac{1}{n}}.
\end{align*}
Thus, 
\begin{align*}
    b 
        &= \frac{\PP_{v}\SqBr{\text{$\not\exists$ high}} \Paren{1 - \Paren{1 - \PP\cSqBr{\text{low}}{\text{not high}}}^n}}{n\PP\SqBr{\text{low}}} \\
        &= \frac{\Omega\Paren{\Paren{1 - \frac{1}{T}}^n} \Paren{1 - \Paren{1 - \Omega\Paren{\frac{1}{n}}}^n}}{nO\Paren{\frac{1}{n}}} \\
        &= \Omega\Paren{\Paren{1 - \frac{1}{T}}^n}.
\end{align*}
\end{proof}

\mainResult*

\begin{proof}
Recall that by Corollary \ref{corollary:NSN-vs-SREV} and Lemma \ref{lem:srev}, $\textsc{SRev}_{n+c}$ exceeds the revenue of the Not-So-Naive Auction only if
\[
    c \geq \frac{bn^2 \Paren{\frac{T}{mn} - \frac{1}{T}} \Paren{1 - \Paren{1 - \frac{1}{T}}^m - \frac{m}{T}\Paren{1 - \frac{1}{T}}^{m-1}}}{\Paren{1 - \frac{1}{T}}^n} .
\]
Note that if $m \in [n/4, n]$, then
\begin{align*}
    c 
        &\geq \frac{bn^2 \Paren{\frac{T}{mn} - \frac{1}{T}} \Paren{1 - e^{-\frac{m}{T}} - \frac{m}{T}e^{-\frac{m-1}{T}}}}{\Paren{1 - \frac{1}{T}}^n} \\
        &\geq \frac{b \frac{n^{3/2}}{\sqrt{m}} \Paren{\lambda - \frac{1}{\lambda}} \Paren{1 - e^{-\frac{1}{\lambda} \sqrt{\frac{m}{n}}} - \frac{1}{\lambda} \sqrt{\frac{m}{n}} e^{-\frac{1}{2\lambda} \sqrt{\frac{m}{n}}}}}{\Paren{1 - \frac{1}{T}}^n} \tag{$T = \lambda \sqrt{mn}$} \\
        &\geq \frac{b n \Paren{\lambda - \frac{1}{\lambda}} \Paren{1 - e^{-\frac{1}{2\lambda}} - \frac{1}{2\lambda} e^{-\frac{1}{4 \lambda}}}}{\Paren{1 - \frac{1}{T}}^n} \tag{$1 - e^{-x} - x e^{-x/2}$ is increasing in $x$, $m \in [n/4, n]$} \\
        &= \frac{\Omega\Paren{\Paren{1 - \frac{1}{T}}^n} n }{\Paren{1 - \frac{1}{T}}^n} \tag{Lemma \ref{lemma:b-lower-bound}} \\
        &= \Omega(\sqrt{mn}). \tag{$m \in [n/4, n]$}
\end{align*}
If $m \leq n/4$, then
\begin{align*}
    c 
        &\geq \frac{bn^2 \Paren{\frac{T}{mn} - \frac{1}{T}} \frac{m(m-1)}{2T^2} \Paren{1 - \frac{m-2}{T}}}{\Paren{1 - \frac{1}{T}}^n} \tag{inclusion-exclusion} \\
        &\geq \frac{b \frac{(m-1)\sqrt{mn}}{2 \lambda m} \Paren{1 - \frac{1}{\lambda^2}} \Paren{1 - \frac{1}{\lambda} \sqrt{\frac{m}{n}}}}{\Paren{1 - \frac{1}{T}}^n} \tag{$T = \lambda \sqrt{mn}$} \\
        &= \frac{b \frac{(m-1)\sqrt{mn}}{2 \lambda m} \Paren{1 - \frac{1}{\lambda^2}} \Paren{1 - \frac{1}{2\lambda}}}{\Paren{1 - \frac{1}{T}}^n} \tag{$m \leq n/4$} \\
        &= \frac{\Omega\Paren{\Paren{1 - \frac{1}{T}}^n} \sqrt{mn} }{\Paren{1 - \frac{1}{T}}^n} \tag{Lemma \ref{lemma:b-lower-bound}} \\
        &= \Omega(\sqrt{mn}).
\end{align*}
\end{proof}

\subsection{Proofs from Section~\ref{sec:ERTreduction}} \label{sec:proofs_ERTreduction}

\redclaimone*
\begin{proof}
    We take two cases. First, if $\vec{v}_i \in R_j$, then 
    $$\bar{\varphi}_j \left( F_j^{-1} \left(Q_j^x (\vec{v}) \right) \right) = \bar{\varphi}_j \left( F_j^{-1} \left( F_j (v_{ij}) \right) \right) = \bar{\varphi}_j (v_{ij}) = \Phi_{ij}^\lambda(\vec{v}_i) = CDW_j^x (\vec{v}). $$

    Now, suppose $\vec{v}_i \notin R_j$. Taking expectations over the random sampling of $Q_j^x$ gives 
    \begin{align*}
        \EE_{r} \left[ \bar{\varphi}_j \left( F_j^{-1} \left(Q_j^x (\vec{v}) \right) \right) \right] &= \EE_{q \leftarrow U[q_j(\vec{v}_i), 1]} \left[ \bar{\varphi}_j \left( F_j^{-1} (q) \right) \right] \\
        &= \EE_{v \leftarrow D_j |_{\ge v_{ij}}} \left[ \bar{\varphi}_j (v)\right] \\
        &\ge \EE_{v \leftarrow D_j |_{\ge v_{ij}}} \left[ \varphi_j (v)\right] \\
        &= v_{ij} = CDW_j^x (\vec{v}).
    \end{align*}

    Considering the two cases in combination establishes the desired inequality.

    Further, note that equality holds if and only if $D$ is a regular distribution (there is only one step where an inequality is invoked, and equality of unironed and ironed virtual values is equivalent to regularity).
\end{proof}

\redclaimtwo*
\begin{proof}
    We have 
    \begin{align*}
        \srev_{n+c}(D) &= \EE_{\vec{v} \leftarrow D^{n+c}} \left[\sum_j \bar{\varphi}_j \left(F_j^{-1} \left(S_{n+c}(\vec{v}) \right) \right) \right] \\
        &\ge \EE_{\vec{v} \leftarrow D^{n}, r} \left[\sum_j \bar{\varphi}_j \left(F_j^{-1} \left(Q_j^x(\vec{v}) \right) \right) \right] \\
        &\ge \EE_{\vec{v} \leftarrow D^{n}} \left[\sum_j CDW_j^x(\vec{v}) \right], 
    \end{align*}
    where the first line is just the definition of $\srev$, the second line follows from stochastic dominance (since $\bar{\varphi}$ and $F_j^{-1}$ are monotone increasing), and the third line is a direct application of Proposition~\ref{prop:reduction_claim1}.
\end{proof}

\redclaimthree*
\begin{proof}
    Fix $q \in [0,1)$ and consider the equal revenue distribution truncated at $\frac{1}{1-q}$. Using Observation~\ref{obs:ERTquantiles}, we have
    \begin{align*}
        \srev_{n+c}(\ERT[\frac{1}{1-q}]) &= \EE_{\vec{v} \leftarrow \left(\ERT[\frac{1}{1-q}]\right)^{n+c}} \left[\sum_j \bar{\varphi}_j \left(F_j^{-1} \left(S_{n+c}(\vec{v}) \right) \right) \right] \\
        &= \sum_j 1 \cdot \mathbb{P} \left(F_j \left(F_j^{-1} \left(S_{n+c}(\vec{v}) \right) \right) \ge q\right) \\
        &= \sum_j \mathbb{P} \left( S_{n+c}(\vec{v}) \ge q\right).
    \end{align*}
    
    Also, since $\ERT[\frac{1}{1-q}]$ is regular, we have
    \begin{align*}
        \EE_{\vec{v} \leftarrow \left(\ERT[\frac{1}{1-q}]\right)^{n}} \left[ CDW_j^x (\vec{v}) \right] &= \EE_{\vec{v} \leftarrow \left(\ERT[\frac{1}{1-q}]\right)^{n}} \left[\bar{\varphi}_j \left(F_j^{-1} \left(Q_j^x (\vec{v}) \right) \right) \right] \\
        &= 1 \cdot \mathbb{P} \left(F_j \left(F_j^{-1} \left(Q_j^x(\vec{v}) \right) \right) \ge q\right) \\
        &= \mathbb{P} \left( Q_j^x(\vec{v}) \ge q\right).
    \end{align*}
    
    Now, observe that 
    \begin{align*}
        S_{n+c} \succsim Q_j^x(\vec{v}) & \iff \mathbb{P} \left(S_{n+c} < q \right) \le \mathbb{P} \left(Q_j^x(\vec{v}) < q \right) \, \forall q\in [0,1] \\
        &\iff \srev_{n+c}(\ERT[\frac{1}{1-q}]) \ge \EE_{\vec{v} \leftarrow \left(\ERT[\frac{1}{1-q}]\right)^{n}} \left[ \sum_j CDW_j^x (\vec{v}) \right] \, \forall q\in [0,1] .
    \end{align*}

    Taking $T = \frac{1}{1-q} \in [1, \infty)$ establishes the claim.
\end{proof}

\reductioncc*
\begin{proof}
    Combining Propositions~\ref{prop:reduction_claim3} and~\ref{prop:reduction_claim2}, we have $\srev_{n+c}(D) \ge \EE_{\vec{v} \leftarrow D^n} \left[\sum_j CDW_j^x (\vec{v}) \right] \ge \rev^{(x,p)} (D^n)$. Taking the supremum over all mechanisms $(x,p)$ gives $\srev_{n+c}(D) \ge \rev(D^n)$, so the competition complexity for $D$ is at most $c$.
\end{proof}

\subsection{Proofs from Section~\ref{sec:logn_m=2}} \label{sec:proofs_logn_m=2}

\ERmotivationone*
\begin{proof}
    
    We evaluate both sides of the Border constraint for types with $v_2 \ge \ln^2 n$.

    The LHS is the expected number of item $2$ awarded to bidders with $v_2 \ge \ln^2 n$, which is
    \begin{align*}
        &n \int_{v_2 = \ln^2 n}^\infty \left(\int_{v_1 = v_2}^{5n} \frac{v_1}{10n} f(v_1 \mid v_1 \ge v_2) dv_1 + \int_{v_1 = 5n}^\infty \frac{1}{2} f(v_1 \mid v_1 \ge v_2) dv_1 \right) f_N(v_2) dv_2 \\
        &= n \int_{v_2 = \ln^2 n}^\infty \left(\int_{v_1 = v_2}^{5n} \frac{v_1}{10n} \cdot \frac{v_2}{v_1^2} dv_1 + \int_{v_1 = 5n}^\infty \frac{1}{2} \cdot \frac{v_2}{v_1^2} dv_1 \right) \frac{2}{v_2^3} dv_2 \\
        &= n \int_{v_2 = \ln^2 n}^\infty \left(\int_{v_1 = v_2}^{5n} \frac{1}{5n v_1}  dv_1 + \int_{v_1 = 5n}^\infty \frac{1}{v_1^2} dv_1 \right) \frac{1}{v_2^2} dv_2 \\
        &= n \int_{v_2 = \ln^2 n}^\infty \left(\frac{1}{5n} \ln\left(\frac{5n}{v_2} \right) + \frac{1}{5n} \right) \frac{1}{v_2^2} dv_2 \\
        &= \frac{1}{5} \int_{v_2 = \ln^2 n}^\infty \frac{\ln(5n/v_2) + 1}{v_2^2} dv_2 \\
        &=  \frac{\ln (5n) - \ln (\ln^2 n))}{5 \ln^2 n} \\
        &\le \frac{\ln(5n)}{5 \ln^2 n} \\
        &\stackrel{n\to\infty}{\to} \frac{1}{5\ln n}.
    \end{align*}
    
    The RHS is the probability that at least one of the $n$ bidders is a type with $v_1 > v_2 \ge \ln^2 n$, which is 
    \begin{align*}
        \mathbb{P}(\exists \text{ a bidder with } v_1 > v_2 \ge \ln^2 n) &= 1 - \mathbb{P}(\text{single bidder does not have } v_1 > v_2 \ge \ln^2 n)^n \\
        &= 1 - (1-\mathbb{P}(\text{single bidder has } v_1 > v_2 \ge \ln^2 n))^n \\
        &= 1 - \left(1 - \frac{1 - F_N(\ln^2 n)}{2}  \right)^n \\
        &= 1 - \left(1 - \frac{1}{2 \ln^4 n} \right)^n \\
        &\sim 1 - \exp\left(-\frac{n}{2 \ln^4 n} \right) \\
        &\stackrel{n\to\infty}{\to} 1,
    \end{align*}
    so feasibility holds.    
\end{proof}

\ERmotivationtwo*
\begin{proof}
    The expected virtual welfare from the non-favorite item for bidders with $v_N \ge \ln^2 n$ is
    \begin{align*}
        &n \int_{v_N = \ln^2 n}^\infty v_N \left(\int_{v_F = v_N}^{5n} \frac{v_F}{10n} f(v_F \mid v_F \ge v_N) dv_F + \int_{v_F = 5n}^\infty \frac{1}{2} f(v_F \mid v_F \ge v_N) dv_F \right) f_N(v_N) dv_N \\
        &\ge n \int_{v_N = \ln^2 n}^{5n} v_N \int_{v_F = v_N}^{5n} \frac{v_F}{10n} \cdot \frac{v_N}{v_F^2} dv_F \cdot \frac{2}{v_N^3} dv_N \\
        &= \int_{v_N = \ln^2 n}^{5n} \frac{1}{5v_N} \int_{v_F = v_N}^{5n} \frac{1}{v_F} dv_F dv_N \\
        &= \int_{v_N = \ln^2 n}^{5n} \frac{1}{5v_N} \ln \left(\frac{5n}{v_N} \right) dv_N \\
        &= \left. -\frac{\ln^2(5n/v_N)}{10} \right|_{v_N = \ln^2 n}^{5n} \\
        &= \frac{\ln^2(5n/\ln^2 n)}{10} \\
        &= \Omega(\ln^2 n).
    \end{align*}
\end{proof}

\ERmotivationthree*
\begin{proof}
    The expected number of each item allocated to a single bidder is
    \begin{align*}
        \int_{v_1 = 1}^\infty \min \left\{\frac{1}{2} , \frac{v_1}{10n} \right\} f(v_1) dv_1 &= \int_{v_1 = 1}^{5n} \frac{v_1}{10n} f(v_1) dv_1 + \int_{v_1 = 5n}^\infty \frac{1}{2}  f(v_1) dv_1 \\
        &\ge \int_{v_1 = 1}^{5n} \frac{v_1}{10n} \cdot \frac{1}{v_1^2} dv_1 \\
        &= \frac{\ln (5n)}{10n} \\
        &\gg \frac{1}{n}.
    \end{align*}
\end{proof}

\mattlem*
\begin{proof}
    Let $q' \coloneqq \pi_1(v_1, v_2)$ denote the probability that $(v_1, v_2)$ gets item $1$. By IR, $p(v_1, v_2) \le v_1 q' + v_2 q$.

    The utility of $(3v_1, v_2')$ for what $(v_1, v_2)$ gets is
    $$3v_1 q' + v_2' q - p(v_1, v_2) \ge 3v_1 q' + v_2' q - v_1 q' - v_2 q = 2v_1 q' + v_2' q - v_2 q.$$
    
    By BIC the utility of $(3v_1, v_2')$ for their own outcome is also at least this amount, so
    \begin{align*}
        3v_1 \pi_1 (3v_1, v_2') + v_2' \pi_2(3v_1, v_2') &\ge 2v_1 q' + v_2' q - v_2 q + p(3v_1, v_2') \\
        &\ge 2v_1 q' + v_2' q - v_2 q.
    \end{align*}

    If $\pi_1(3v_1, v_2') = p \ge \pi_2(3v_1, v_2')$, then this gives
    \begin{align*}
        3v_1 p + v_2' p &\ge 2v_1 q' + v_2' q - v_2 q & \\
        p &\ge \frac{2v_1 q' + v_2' q - v_2 q}{3v_1 + v_2'} & \\
        &\ge \frac{2v_1 q' + 0 - v_1 q}{4v_1} & (v_2' \ge 0, v_2 \le v_1) & \\
        &\ge \frac{2v_1 q' - v_1 q}{4v_1} & (v_2' \le v_2 \le v_1) \\
        &\ge \frac{2v_1 q - v_1 q}{4v_1} & (q' \ge q) \\
        \implies \pi_1(3v_1, v_2') &\ge \frac{q}{4}. &
    \end{align*}
\end{proof}

\mattcor*
\begin{proof}
    By feasibility, the expected allocation of item $1$ awarded to a bidder can be at most $\frac{1}{n/2} = \frac{2}{n}$ (the factor of $2$ comes from assuming $v_1 > v_2$, by item symmetry). We can write this expectation as
    \begin{align*}
        \frac{2}{n} &\ge \EE_{(v_1, v_2)} \left[\pi_1(v_1, v_2) \right] \\
        &= \EE_{v_2} \EE_{v_1} \left[ \pi_1(v_1, v_2) \mid v_2 \right] \\ 
        &= \int_{w = 1}^\infty \mathbb{P}(v_2 = w) \int_{v = w}^\infty \pi_1(v,w) \mathbb{P}(v_1 = v \mid v_2 = w) dv \, dw \\
        &= \int_{w = 1}^\infty \frac{2}{w^3} \int_{v = w}^\infty \pi_1(v,w) \frac{w}{v^2} dv \, dw \\
        &\ge \int_{w = 1}^y \frac{2}{w^2} \int_{v = w}^\infty \frac{\pi_1(v,w)}{v^2} dv \, dw,
    \end{align*}
    where the final inequality is valid since the integrand is positive. Now that we have $w \le y < 3y$, we can split the inner integral and continue as follows: 
    \begin{align*}
        \frac{2}{n} &\ge \int_{w = 1}^y \frac{2}{w^2} \left(\int_{v = w}^{3y} \frac{\pi_1(v,w)}{v^2} dv \, dw + \int_{v = 3y}^\infty \frac{\pi_1(v,w)}{v^2} dv \, dw \right) \\
        &\ge \int_{w = 1}^y \frac{2}{w^2} \int_{v = 3y}^\infty \frac{\pi_1(v,w)}{v^2} dv \, dw.
    \end{align*}

    In this region of the type space, we have $v\ge 3y$ and $w\le y$. By Lemma~\ref{lem:MattLemma}, we have $\pi_1(v, w) \ge \frac{\pi_2(v/3, y)}{4}$, so now 
    \begin{align*}
        \frac{2}{n} &\ge \int_{w = 1}^y \frac{2}{w^2} \int_{v = 3y}^\infty \frac{\pi_2(v/3,y)}{4v^2} dv \, dw \\
        &= \int_{v = 3y}^\infty \frac{\pi_2(v/3,y)}{2v^2} dv \int_{w=1}^y \frac{1}{w^2} dw \\
        &= \left(1-\frac{1}{y}\right) \int_{v = 3y}^\infty \frac{\pi_2(v/3,y)}{2v^2} dv \\
        &= \left(1-\frac{1}{y}\right) \int_{u = y}^\infty \frac{\pi_2(u,y)}{18 u^2} 3 du \\ 
        &= \frac{1-1/y}{6y} \int_{u=y}^\infty \pi_2(u,y) \frac{y}{u^2} du \\
        &= \frac{1-1/y}{6y} \EE_{v_1} \left[\pi_2(v_1, v_2) \mid v_2 = y \right].
    \end{align*}
    Rearranging and using $y\ge 2$, we get
    $$\EE_{v_1} \left[\pi_2(v_1, v_2) \mid v_2 = y \right] \le \frac{12y}{n(1-1/y)} \le \frac{12y}{n/2} = \frac{24y}{n}.$$
\end{proof}

\subsection{Calculations for $\ER^m$}
We now provide some computations that will be useful for the rest of the analysis of $\ER^m$:
\begin{itemize}
    \item Distribution of favorite item:
    \begin{align*}
        F_{F}(x) &= \left(1 - \frac1x \right)^m, \\
        f_F (x) &= m \left(1-\frac1x \right)^{m-1} \frac{1}{x^2}
    \end{align*}

    \item Distribution of the $m-1$ non-favorite items:
    \begin{align*}
        f_{NF}(x) &= \mathbb{P}(v_j = x \mid \text{item } j \text{ is a non-fav}) \\
        &= \frac{ \mathbb{P}(v_j = x \cap \text{item } j \text{ is a non-fav})}{ \mathbb{P}(\text{item } j \text{ is a non-fav})} \\
        &= \frac{\frac{1}{x^2} \mathbb{P} (\text{one of the other } m-1 \text{ values is} > x)}{1 - \frac1m} \\
        &= \frac{\frac{1}{x^2} \left(1 - (1-\frac1x)^{m-1} \right)}{1 - \frac1m}
    \end{align*}

    \item Non-favorite marginal:
    \[
        F_{NF}(v) = \int_1^v \frac{\frac{1}{x^2} \left(1 - (1-\frac1x)^{m-1} \right)}{1 - \frac1m} \,\mathrm{d}v = \frac{1}{1-\frac1m}\Paren{1 - \frac{1}{v} - \frac{1}{m}\Paren{1-\frac{1}{v}}^m}
    \]
\end{itemize}

\subsection{Proofs from Section~\ref{sec:ER-revenue-UB}} \label{sec:proofs_ER-revenue-UB}

\refinedMatt*
\begin{proof}
By BIC, the utility of a bidder with value $v$ satisfies
\[
    \sum_j v_j \pi_j(v) - p(v) \geq \sum_j v_j \pi_j(\lambda v_k, v_\ell', v_{-\{k, \ell\}}) - p(\lambda v_k, v_\ell', v_{-\{k, \ell\}}).
\]
Similarly, the utility of a bidder with value $(\lambda v_k, v_\ell', v_{-\{k, \ell\}})$ satisfies
\begin{align*}
    \lambda v_k \pi_k(v) & + v'_\ell \pi_\ell(v) + \sum_{j \not= k, \ell} v_j \pi_j(v) - p(v) \\
        &\leq \lambda v_k \pi_k(\lambda v_k, v_\ell', v_{-\{k, \ell\}}) + v'_\ell \pi_\ell (\lambda v_k, v_\ell', v_{-\{k, \ell\}}) + \sum_{j \not= k, \ell} v_j \pi_j(\lambda v_k, v_\ell', v_{-\{k, \ell\}}) - p(\lambda v_k, v_\ell', v_{-\{k, \ell\}}).
\end{align*}
Subtracting the first inequality from the second yields
\[
    (\lambda - 1) v_k \pi_k(v) + (v'_\ell - v_\ell) \pi_\ell(v) \leq (\lambda - 1) v_k \pi_k(\lambda v_k, v_\ell', v_{-\{k, \ell\}}) + (v'_\ell - v_\ell) \pi_\ell(\lambda v_k, v_\ell', v_{-\{k, \ell\}}).
\]
Since $v_\ell' \leq v_\ell$, we can further upper bound the RHS by getting rid of the term associated with item $\ell$:
\[
    (\lambda - 1) v_k \pi_k(v) + (v'_\ell - v_\ell) \pi_\ell(v) \leq (\lambda - 1) v_k \pi_k(\lambda v_k, v_\ell', v_{-\{k, \ell\}}).
\]
Since $v_k \geq v_\ell$, by strong monotonicity, $\pi_k(v) \geq \pi_\ell(v)$, so we can further lower bound $(v'_\ell - v_\ell) \pi_\ell(v)$ by $- v_k \pi_k(v)$.
Thus,
\[
    (\lambda - 2) v_k \pi_k(v) \leq (\lambda - 1) v_k \pi_k(\lambda v_k, v_\ell', v_{-\{k,\ell\}}).
\]
The result now follows from dividing by $(\lambda - 1) v_k$ (which is well-defined and preserves the direction of the inequality since $\lambda > 1$) and applying strong monotonicity.
\end{proof}

\jointpdfcond*
\begin{proof}
Note that
\[
    \PP\cSqBr{v_k \leq w_k}{v_\ell, E_k} = \PP\cSqBr{\max_{j \not = \ell} v_j \leq w_k}{v_\ell, \max_{j \not = \ell} v_j \geq v_\ell} = \frac{\Paren{1 - \frac{1}{w_k}}^{m-1} - \Paren{1 - \frac{1}{v_\ell}}^{m-1}}{1 - \Paren{1 - \frac{1}{v_\ell}}^{m-1}} \cdot \1(v_\ell \leq w_k),
\]
so the pdf of $v_k$ conditioned on $v_\ell$ and $E_k$ is 
\[
    \frac{\partial}{\partial w_k} \PP\cSqBr{v_k \leq w_k}{v_\ell, E_k} = \frac{(m-1)\Paren{1 - \frac{1}{w_k}}^{m-2} \frac{1}{w_k^2}}{1 - \Paren{1 - \frac{1}{v_\ell}}^{m-1}} \cdot \1(v_\ell \leq w_k).
\]
Now, by independence (conditioned on the maximum, the non-maximum random variables become i.i.d.), 
\[
    \PP\cSqBr{v_j \leq w_j \; \forall j \not= \ell}{v_k, v_\ell, E_k} = \1(v_k \leq w_k) \prod_{j \not= k, \ell} \min\CrBr{\frac{1 - \frac{1}{w_j}}{1 - \frac{1}{v}}, 1}.
\]
Integrating the above expression over the pdf of $v_k$ conditioned on $v_\ell$ and $E_k$ yields the result. 
To obtain the joint pdf, observe that when differentiating with respect to each variable $w_j$, if $w_j \geq w_k$ for some $j \not= k, \ell$, then no $w_j$ term appears in the joint cdf so the partial becomes $0$.
\end{proof}

\expNFprob*
\begin{proof}
The pdf of $v_\ell$ conditioned on $E_k$ is 
\[
    \frac{\frac{1}{w_\ell^2} \left(1 - (1-\frac{1}{w_\ell})^{m-1} \right)}{1 - \frac1m}
\]
We compute:
\begin{align*}
    \EE\cSqBr{\pi_k(v)}{E_k}
        &= \EE\cSqBr{\EE\cSqBr{\pi_k(v)}{v_\ell, E_k}}{E_k} \\
        &= \int_1^\infty \int_{w_\ell}^\infty\int_1^{w_k} \dots \int_1^{w_k} \pi_k(w) \frac{m-1}{1-\frac{1}{m}} \prod_j \frac{1}{w_j^2} \, \mathrm{d}w \tag{Lemma \ref{lem:joint-pdf-given-non-fav}; integrate over the pdf of $v_\ell$ conditioned on $E_k$} \\
        &\geq \int_1^y \int_{3y}^\infty\int_1^{\frac{w_k}{3}} \dots \int_1^{\frac{w_k}{3}} \pi_k(w) \frac{m-1}{1-\frac{1}{m}} \prod_j \frac{1}{w_j^2} \, \mathrm{d}w \\
        &\geq \int_1^y \int_{3y}^\infty\int_1^{\frac{w_k}{3}} \dots \int_1^{\frac{w_k}{3}} \frac{\pi_\ell\Paren{\frac{w_k}{3}, y, w_{-\{k, \ell\}}}}{2} \frac{m-1}{1-\frac{1}{m}} \prod_j \frac{1}{w_j^2} \, \mathrm{d}w \tag{Lemma \ref{lem:refined-MattLemma_m} with $\lambda = 3$} \\
        &= \frac{1}{6} \int_1^y \int_y^\infty\int_1^{w_k} \dots \int_1^{w_k} \pi_\ell\Paren{y, w_{-\ell}} \frac{m-1}{1-\frac{1}{m}} \prod_j \frac{1}{w_j^2} \, \mathrm{d}w \tag{change of variable} \\
        &= \Paren{1 - \frac{1}{y}} \int_y^\infty\int_1^{w_k} \dots \int_1^{w_k} \pi_\ell\Paren{y, w_{-\ell}} \frac{m-1}{1 - \frac{1}{m}} \prod_{j \not= \ell} \frac{1}{w_j^2} \, \mathrm{d}w_{-\ell} \\
        &= \frac{\Paren{1 - \Paren{1 - \frac{1}{y}}^{m-1}}\Paren{1 - \frac{1}{y}}}{6\Paren{1 - \frac{1}{m}}} \int_y^\infty\int_1^{w_k} \dots \int_1^{w_k} \pi_\ell\Paren{y, w_{-\ell}} \frac{m-1}{1 - \Paren{1 - \frac{1}{y}}^{m-1}} \prod_{j \not= \ell} \frac{1}{w_j^2} \, \mathrm{d}w_{-\ell} \\
        &= \frac{\Paren{1 - \Paren{1 - \frac{1}{y}}^{m-1}}\Paren{1 - \frac{1}{y}}}{6\Paren{1 - \frac{1}{m}}} \cdot \EE\cSqBr{\pi_\ell(y, v_{-\ell})}{v_\ell = y, E_k} \tag{Lemma \ref{lem:joint-pdf-given-non-fav}}
\end{align*}
\end{proof}

\ubNFprob*
\begin{proof}
We have
\begin{align*}
    \frac{m}{n} 
        &\geq \EE\SqBr{\sum_{k = 1}^m \pi_k(v) \cdot \1(E_k)} \\
        &= \frac{1}{m} \sum_{k = 1}^m \EE\cSqBr{\pi_k(v)}{E_k} \\
        &\geq \frac{1}{m} \sum_{k \not= j} \EE\cSqBr{\pi_k(v)}{E_k} \\
        &\geq \frac{\Paren{1 - \Paren{1 - \frac{1}{y}}^{m-1}}\Paren{1 - \frac{1}{y}}}{6(m-1)} \sum_{k \not= j} \EE\cSqBr{\pi_j(y, v_{-j})}{v_j = y, E_k} \tag{Lemma \ref{lem:expected-prob-of-getting-non-fav-via-expected-prob-of-getting-particular-fav-improved}} \\
        &= \frac{1}{6} \Paren{1 - \Paren{1 - \frac{1}{y}}^{m-1}}\Paren{1 - \frac{1}{y}} \EE\cSqBr{\pi_j(y, v_{-j})}{v_j = y, E_{-j}} 
\end{align*}
\end{proof}

\contributionNF*
\begin{proof}
We have
\begin{align*}
    \EE\cSqBr{v_j \pi_j(v)}{E_{-j}}
        = {} & \int_1^\infty y\EE\cSqBr{\pi_j(y, v_{-j})}{v_j = y, E_{-j}} \frac{\frac{1}{y^2}\Paren{1 - \Paren{1 - \frac{1}{y}}^{m-1}}}{1 - \frac{1}{m}} \,\mathrm{d}y \\
        \leq {} & \Paren{1 + \frac{1}{n}} \int_1^{1 + \frac{1}{n}} \frac{\frac{1}{y^2} \Paren{1 - \Paren{1 - \frac{1}{y}}^{m-1}}}{1 - \frac{1}{m}} \,\mathrm{d}y + \int_{1 + \frac{1}{n}}^{n+1} \frac{6m}{n(y-1)\Paren{1 - \frac{1}{m}}} \,\mathrm{d}y \\
            & + \int_{n+1}^\infty \frac{\frac{1}{y} \Paren{1 - \Paren{1 - \frac{1}{y}}^{m-1}}}{1 - \frac{1}{m}} \,\mathrm{d}y \tag{Lemma \ref{lem:upper-bound-expected-prob-non-fav-given-val-improved}} \\
        \leq {} & \frac{1 + \frac{1}{n}}{1 - \frac{1}{m}} \int_1^{1 + \frac{1}{n}} \frac{m}{y^2} \,\mathrm{d}v + \frac{12m \ln n}{n\Paren{1 - \frac{1}{m}}} + \frac{1}{1 - \frac{1}{m}} \int_{n+1}^\infty \frac{m}{y^2} \,\mathrm{d}v \tag{union bound} \\
        \leq {} & O\Paren{\frac{m \ln n}{n}} + O\Paren{\frac{m}{n}}.
\end{align*}
\end{proof}

\revERmn*
\begin{proof}
\begin{align*}
    \rev_n(\ER^m)
        &\leq \max_{\text{$(\pi, p)$ BIC}} \underset{v \sim (\ER^m)^n}{\EE}\SqBr{\sum_i \sum_j \pi_{ij}(v_i) \Paren{\varphi_{ij}(v_i) \cdot \1(E_j) + v_{ij} \cdot \1(E_{-j})}} \tag{\cite{CaiDW16}, \cite{BeyhaghiW19}} \\
        &\leq nm + \max_{\text{$(\pi, p)$ BIC}} \underset{v \sim (\ER^m)^n}{\EE}\SqBr{\sum_i \sum_j v_{ij} \pi_{ij}(v_i) \cdot \1(E_{-j})} \tag{\cite{BeyhaghiW19}} \\
        &\leq nm + \sum_i \sum_j O\Paren{\frac{m \ln n}{n}} \tag{Lemma \ref{lemma:contribution-NF-virtual-welfare}} \\
        &= nm + O(m^2 \ln n)
\end{align*}
\end{proof}

\subsection{Proofs from Section~\ref{sec:ER-revenue-LB}} \label{sec:proofs_ER-revenue-LB}

\ESecondHighestFav*
\begin{proof}
To see the upper bound, simply note that $v_{(2),(1)}$ is at most the second highest value of $nm$ draws from $\ER$, which has expectation $nm$.
To see the lower bound, note that for any bidder $i$,
\[
    \PP\SqBr{v_{i, (1)} \leq z} = \PP\SqBr{\max_{j \in [m]} v_{i,j} \leq z} = \Paren{1-\frac{1}{z}}^m
\]
Thus,
\begin{align*}
    \PP\SqBr{v_{(2), (1)} \leq z} 
        &= \Paren{1-\frac{1}{z}}^{mn} + n\Paren{1-\frac{1}{z}}^{m(n-1)}\Paren{1-\Paren{1-\frac{1}{z}}^m} \\
        &= \Paren{1-\frac{1}{z}}^{mn} + \frac{mn}{z}\Paren{1-\frac{1}{z}}^{mn-1} + n \Paren{1- \Paren{1+\frac{m}{z-1}}\Paren{1-\frac{1}{z}}^m} \Paren{1-\frac{1}{z}}^{m(n-1)} \\
        &\leq \Paren{1-\frac{1}{z}}^{mn} + \frac{mn}{z}\Paren{1-\frac{1}{z}}^{mn-1} + \frac{nm(m-1)}{z(z-1)} \Paren{1-\frac{1}{z}}^{m(n-1)}
\end{align*}
and
\begin{align*}
    \EE\SqBr{v_{(2),(1)}}
        &= 1 + \int_1^\infty \PP\SqBr{v_{(2), (1)} \geq z} \,\mathrm{d}z \\
        &\geq 1 + \int_1^\infty 1 - \Paren{1-\frac{1}{z}}^{mn} - \frac{mn}{z}\Paren{1-\frac{1}{z}}^{mn-1} \,\mathrm{d}z - \int_1^\infty \frac{nm(m-1)}{z(z-1)} \Paren{1-\frac{1}{z}}^{m(n-1)} \,\mathrm{d}z \\
        &= 1 + \int_1^\infty \PP[\text{second highest value among $mn$ i.i.d. draws} \geq z] \,\mathrm{d}z - \frac{n(m-1)}{n-1} \tag{Wolfram} \\
        &= \EE[\text{second highest value among $mn$ i.i.d. draws}] - \frac{n(m-1)}{n-1} \\
        &= mn - \frac{n(m-1)}{n-1}
\end{align*}
\end{proof}

\cENonFav*
\begin{proof}
$x$ is a random draw from $\ER$ that is at most $v$, so
\[
    \PP\cSqBr{x \leq z}{x \leq v} = \begin{cases} 
        0 & z < 1 \\
        \frac{1-\frac{1}{z}}{1 - \frac{1}{v}} & 1 \leq z \leq v \\
        1 & z > v
    \end{cases}
\]
Thus,
\begin{align*}
    \EE\cSqBr{x}{x \leq v} 
        &= 1 + \int_1^\infty \PP\cSqBr{x \geq z}{x \leq v} \,\mathrm{d}z \\
        &= 1 + \int_1^v 1 - \frac{1-\frac{1}{z}}{1 - \frac{1}{v}} \,\mathrm{d}z \\
        &= \frac{\ln v}{1 - 1/v}
\end{align*}
\end{proof}

\begin{lemma}\label{lemma:conditional-variance-ER}
Let $x \sim \ER$.
\[
    \Var\cSqBr{x}{x \leq v} = v - \Paren{\frac{\ln v}{1 - 1/v}}^2
\]
\end{lemma}

\begin{proof}
By Lemma \ref{lemma:conditional-expectation-non-favorite}, it suffices to compute $\EE\cSqBr{x^2}{x \leq v}$.
\[
    \EE\cSqBr{x^2}{x \leq v} = \int_1^v z^2 \frac{\mathrm{d}}{\mathrm{d}z} \PP\cSqBr{x \leq z}{x \leq v} \, \mathrm{d}z = \int_1^v \frac{1}{1-\frac{1}{v}} \, \mathrm{d}z = v
\]
\end{proof}

\begin{theorem}\label{theorem:brev-ER-LB}
$$\rev^{\text{SPA-GB}}(\ER^m) \geq nm + \Omega(m \ln(nm))$$
\end{theorem}

\begin{proof}
Note that 
\begin{align*}
    \rev^{\text{SPA-GB}}(\ER^m) 
        &\geq \EE\SqBr{\sum_{j=1}^m v_{(2), (j)} \cdot \1\Paren{\textstyle v_{(1), (1)} \geq \sum_{j=1}^m v_{(2), (j)}} + v_{(2),(1)} \cdot \1\Paren{\textstyle v_{(1), (1)} < \sum_{j=1}^m v_{(2), (j)}}} \\
        &= \EE\SqBr{v_{(2),(1)} + \sum_{j=2}^m v_{(2),(j)} \cdot \1\Paren{\textstyle v_{(1), (1)} < \sum_{j=1}^m v_{(2), (j)}}} \\
        &\geq nm - O(m) + \EE\SqBr{\sum_{j=2}^m v_{(2),(j)} \cdot \1\Paren{\textstyle v_{(1), (1)} < \sum_{j=1}^m v_{(2), (j)}}} \tag{Lemma \ref{lemma:expectation-second-highest-favorite}}
\end{align*}
since if bidder (1) values her favorite item at least as much as bidder (2) values the bundle, then the second highest value for the bundle is at least bidder (2)'s value and if bidder (2) values the bundle more than bidder (1) values her favorite item, then the second highest value for the bundle is at least bidder (1)'s value for her favorite item (which in turn is at least bidder (2)'s value for her favorite item).
It remains to lower bound the expectation.

Let $E$ denote the event that $v_{(2),(j)} \leq v_{(2),(1)}^{1/2}$ for all $j \in [m] \setminus \{1\}$.
\begin{align*}
    & \EE \cSqBr{\sum_{j=2}^m v_{(2),(j)} \cdot \1\Paren{\textstyle v_{(1), (1)} \geq \sum_{j=1}^m v_{(2), (j)}}}{v_{(2),(1)}} \\
        &  \geq \EE\cSqBr{\sum_{j=2}^m v_{(2),(j)} \cdot \1\Paren{\textstyle v_{(1), (1)} \geq \sum_{j=1}^m v_{(2), (j)}}}{v_{(2),(1)}, E} \PP\cSqBr{E}{v_{(2),(1)}} \\
        &  \geq \EE\cSqBr{\sum_{j=2}^m v_{(2),(j)} \cdot \1\Paren{v_{(1), (1)} \geq v_{(2), (1)} + (m-1)v_{(2),(1)}^{1/2}}}{v_{(2),(1)}, E} \PP\cSqBr{E}{v_{(2),(1)}} \\
        &  \geq \EE\cSqBr{\sum_{j=2}^m v_{(2),(j)}}{v_{(2),(1)}, E} \PP\cSqBr{v_{(1), (1)} \geq v_{(2), (1)} + (m-1)v_{(2),(1)}^{1/2}}{v_{(2),(1)}} \PP\cSqBr{E}{v_{(2),(1)}} \tag{$\sum_{j=2}^m v_{(2),(j)}$ and $\1\Paren{v_{(1), (1)} \geq v_{(2), (1)} + (m-1)v_{(2),(1)}^\varepsilon}$ are independent conditioned on $v_{(2),(1)}$}
\end{align*}
Since $v_{(2),(2)}, \dots, v_{(2),(m)}$ consist of $m-1$ i.i.d. random draws from $\ER$ conditioned on being at most $v_{(2),(1)}$, 
\begin{align*}
    \PP\cSqBr{E}{v_{(2),(1)}} 
        &= \Paren{\frac{1-\frac{1}{v_{(2),(1)}^\varepsilon}}{1 - \frac{1}{v_{(2),(1)}}}}^{m-1} \geq \Paren{1 - \frac{1}{v_{(2),(1)}^{1/2}}}^{m-1} \geq 1 - \frac{m-1}{v_{(2),(1)}^{1/2}}
\end{align*}
Meanwhile, 
\begin{align*}
    & \PP \cSqBr{v_{(1), (1)} \geq v_{(2), (1)} + (m-1)v_{(2),(1)}^{1/2}}{v_{(2),(1)}} \\
        &  = \PP_{x \sim \ER^m}\cSqBr{\max_j x_j \geq v_{(2), (1)} + (m-1)v_{(2),(1)}^{1/2}}{v_{(2),(1)}, \max_j x_j \geq v_{(2),(1)}} \tag{conditioned on $v_{(2),(1)}$, $v_{(1),(1)}$ is just the maximum of $m$ i.i.d. random draws from $\ER$} \\
        &  \geq \PP_{x \sim \ER}\cSqBr{v_{i,j} \geq v_{(2), (1)} + (m-1)v_{(2),(1)}^{1/2}}{v_{(2),(1)}, v_{i,j} \geq v_{(2),(1)}} \\
        &  = \frac{v_{(2),(1)}}{v_{(2),(1)} + (m-1)v_{(2),(1)}^{1/2}} \\
        &  = \frac{1}{1 + \frac{m-1}{v_{(2),(1)}^{1/2}}} \\
\end{align*}
Moreover,
\begin{align*}
    \EE\cSqBr{\sum_{j=2}^m v_{(2),(j)}}{v_{(2),(1)}, E}
        &= (m-1) \int_1^{v_{(2),(1)}^{1/2}} 1 - \frac{1-\frac{1}{v}}{1 - \frac{1}{v_{(2),(1)}^{1/2}}} \,\mathrm{d}v \\
        &= \frac{(m-1) \Paren{\frac{\ln v_{(2),(1)}}{2} - 1}}{1 - \frac{1}{v_{(2),(1)}^{1/2}}} \\
        &\geq (m-1) \Paren{\frac{\ln v_{(2),(1)}}{2} - 1}
\end{align*}
Putting all three terms together,
\[
    \EE\cSqBr{\sum_{j=2}^m v_{(2),(j)} \cdot \1\Paren{\textstyle v_{(1), (1)} \geq \sum_j v_{(2), (j)}}}{v_{(2),(1)}} \geq (m-1) \Paren{\frac{\ln v_{(2),(1)}}{2} - 1} \frac{1 - \frac{m-1}{v_{(2),(1)}^{1/2}}}{1 + \frac{m-1}{v_{(2),(1)}^{1/2}}} 
\]
For $v_{(2),(1)} \geq 4(m-1)^2$, the RHS is at least $\frac{m-1}{3}  \Paren{\frac{\ln v_{(2),(1)}}{2} - 1}$.
To conclude,
\begin{align*}
    & \EE \SqBr{\sum_{j=2}^m v_{(2),(j)} \cdot \1\Paren{\textstyle v_{(1), (1)} \geq \sum_j v_{(2), (j)}}} \\
        &  \geq \int_{4(m-1)^2}^\infty \EE\cSqBr{\sum_{j=2}^m v_{(2),(j)} \cdot \1\Paren{\textstyle v_{(1), (1)} \geq \sum_j v_{(2), (j)}}}{v_{(2),(1)}} \frac{\mathrm{d}}{\mathrm{d}v} \PP[v_{(2),(1)} \leq v] \,\mathrm{d}v \\
        &  \geq \frac{m-1}{3}  \int_{4(m-1)^2}^\infty \Paren{\frac{\ln v}{2} - 1} \frac{\mathrm{d}}{\mathrm{d}v} \PP[v_{(2),(1)} \leq v] \,\mathrm{d}v \\
        &  \geq \frac{m-1}{3} \int_{4 nm}^\infty \Paren{\frac{\ln v}{2} - 1} \frac{\mathrm{d}}{\mathrm{d}v} \PP[v_{(2),(1)} \leq v] \,\mathrm{d}v \tag{$n \geq m$} \\
        &  \geq \frac{m-1}{3} \Paren{\frac{\ln(4mn) }{2} - 1} \PP\SqBr{v_{(2),(1)} \geq 4nm}
\end{align*}
Note that
\begin{align*}
    \PP\SqBr{v_{(2),(1)} \geq 4nm} 
        &= 1 - \Paren{1 - \frac{1}{4nm}}^{mn} - n \Paren{1 - \frac{1}{4nm}}^{m(n-1)}\Paren{1 - \Paren{1 - \frac{1}{4nm}}^{m}} \\
        &\geq 1 - e^{-\frac{1}{4}} - \frac{1}{4} e^{-\frac{1}{4}\Paren{1 - \frac{1}{n}}} \\
        &\geq  1 - e^{-1/4} - \frac{1}{4} e^{-1/8} \tag{$n \geq 2$} \\
        &> 0.0005 
\end{align*}
so
\[
    \EE\SqBr{\sum_{j=2}^m v_{(2),(j)} \cdot \1\Paren{\textstyle v_{(1), (1)} > \sum_j v_{(2), (j)}}} \geq \frac{0.0005}{3} (m-1) \Paren{\frac{\ln(4mn) }{2} - 1} = \Omega(m \ln (mn))
\]
\end{proof}

\begin{theorem}\label{theorem:brev-ER-UB}
$$\rev^{\text{SPA-GB}}(\ER^m) \leq nm + O(m \ln(nm))$$.
\end{theorem}

\begin{proof}
Note that
\begin{align*}
    \rev^{\text{SPA-GB}}(\ER^m) 
        &\leq \EE\SqBr{\max_{i \geq 2} \sum_{j=1}^m v_{(i),(j)}} \\
        &= \EE\SqBr{v_{(2),(1)} + \frac{(m-1) \ln v_{(2),(1)}}{1 - \frac{1}{v_{(2),(1)}}} + \Paren{\max_{i \geq 2} \sum_{j=1}^m v_{(i),(j)} - v_{(2),(1)} - \frac{(m-1) \ln v_{(2),(1)}}{1 - \frac{1}{v_{(2),(1)}}}}} \\
        &\leq nm + \frac{(m-1) \ln(nm)}{1 - \frac{1}{nm}} + \EE\SqBr{\max_{i \geq 2} \sum_{j=1}^m v_{(i),(j)} - v_{(2),(1)} - \frac{(m-1) \ln v_{(2),(1)}}{1 - \frac{1}{v_{(2),(1)}}}} \tag{Lemma \ref{lemma:expectation-second-highest-favorite}; Jensen applied to $\frac{\ln x}{1 - 1/x}$} \\
        &\leq nm + \frac{(m-1) \ln(nm)}{1 - \frac{1}{nm}} + \sum_{i=2}^n \EE\SqBr{\Paren{\sum_{j=1}^m v_{(i),(j)} - v_{(2),(1)} - \frac{(m-1) \ln v_{(2),(1)}}{1 - \frac{1}{v_{(2),(1)}}}}^+}
\end{align*}
since $\max_{i \geq 2} \sum_{j=1}^m v_{(i),(j)}$ is either revenue if the bidder with the highest value for any item also has the highest value for the bundle or welfare otherwise.
We show that the sum of expectations on the last line is at most $O(m \ln(mn))$.

Note that conditioned on $v_{(i),(1)}$, for all $j \geq 2$, we can treat $v_{(i),(j)}$ as an independent draw from $\ER$ conditioned on being at most $v_{(i),(1)}$.
Thus, for all $i \geq 3$,
\begin{align*}
    &\PP\cSqBr{\sum_{j = 1}^m v_{(i),(j)} \geq v_{(2),(1)} + \frac{(m-1) \ln v_{(2),(1)}}{1 - \frac{1}{v_{(2),(1)}}} + t(m-1)}{v_{(i),(1)}, v_{(2),(1)}} \\
        &  \leq \frac{(m-1) v_{(i),(1)} \cdot \1\Paren{m v_{(i),(1)} \geq v_{(2),(1)} + \frac{(m-1) \ln v_{(2),(1)}}{1 - \frac{1}{v_{(2),(1)}}} + t(m-1)}}{\Paren{v_{(2),(1)} + \frac{(m-1) \ln v_{(2),(1)}}{1 - \frac{1}{v_{(2),(1)}}} + t(m-1) - \EE\cSqBr{\sum_{j = 1}^m v_{(i),(j)}}{v_{(i),(1)}, v_{(2),(1)}}}^2}  \tag{Lemma \ref{lemma:conditional-variance-ER}; Chebyshev} \\
        &  \leq \frac{(m-1) v_{(i),(1)} \cdot \1\Paren{m v_{(i),(1)} \geq v_{(2),(1)} + \frac{(m-1) \ln v_{(2),(1)}}{1 - \frac{1}{v_{(2),(1)}}} + t(m-1)}}{\Paren{v_{(2),(1)} - v_{(i),(1)} + t(m-1)}^2}  \tag{by Lemma \ref{lemma:conditional-expectation-non-favorite}, $\EE\cSqBr{\sum_{j = 2}^m v_{(i),(j)}}{v_{(i),(1)}, v_{(2),(1)}} = \frac{(m-1) \ln v_{(i),(1)}}{1 - \frac{1}{v_{(i),(1)}}} \leq \frac{(m-1) \ln v_{(2),(1)}}{1 - \frac{1}{v_{(2),(1)}}}$}
\end{align*}
Moreover, conditioned on $v_{(2),(1)}$, for all $i \geq 3$, we can treat $v_{(i),(1)}$ as the maximum of $m$ independent draws from $\ER$ conditioned on being at most $v_{(2),(1)}$.
Note that the cdf and pdf of such a random variable is 
\begin{align*}
    \underset{v \sim \ER^m}{\PP}\cSqBr{\max_j v_j \leq z}{\max_j v_j \leq v_{(2),(1)}} 
        &= \Paren{\frac{1 - \frac{1}{z}}{1 - \frac{1}{v_{(2),(1)}}}}^m \\
    \frac{\mathrm{d}}{\mathrm{d}z} \underset{v \sim \ER^m}{\PP}\cSqBr{\max_j v_j \leq z}{\max_j v_j \leq v_{(2),(1)}}
        &= \frac{m}{z^2\Paren{1 - \frac{1}{v_{(2),(1)}}}}\Paren{\frac{1 - \frac{1}{z}}{1 - \frac{1}{v_{(2),(1)}}}}^{m-1}
\end{align*}
Integrating over the pdf of $v_{(i),(1)}$ treated as the maximum of $m$ independent draws from $\ER$ conditioned on being at most $v_{(2),(1)}$ yields
\begin{align*}
    &\PP\cSqBr{\sum_{j = 1}^m v_{(i),(j)} \geq v_{(2),(1)} + \frac{(m-1) \ln v_{(2),(1)}}{1 - \frac{1}{v_{(2),(1)}}} + t(m-1)}{v_{(2),(1)}} \\
    & \leq (m-1)\int_{\frac{v_{(2),(1)}}{m} + \Paren{1 - \frac{1}{m}}\Paren{\frac{\ln v_{(2),(1)}}{1 - \frac{1}{v_{(2),(1)}}} + t}}^{v_{(2),(1)}} \frac{v}{\Paren{v_{(2),(1)} - v + t(m-1)}^2} \frac{\mathrm{d}}{\mathrm{d}v} \PP\cSqBr{v_{(i),(1)} \leq v}{v_{(2),(1)}} \, \mathrm{d}v \\
    & = \frac{m-1}{1 - \frac{1}{v_{(2),(1)}}}\int_{\frac{v_{(2),(1)}}{m} + \Paren{1 - \frac{1}{m}}\Paren{\frac{\ln v_{(2),(1)}}{1 - \frac{1}{v_{(2),(1)}}} + t}}^{v_{(2),(1)}} \frac{v}{\Paren{v_{(2),(1)} - v + t(m-1)}^2} \frac{m}{v^2} \Paren{\frac{1 - 1/v}{1-1/v_{(2),(1)}}}^{m-1} \, \mathrm{d}v \\
    & \leq \frac{m(m-1)}{1 - \frac{1}{v_{(2),(1)}}}\int_{\frac{v_{(2),(1)}}{m} + \Paren{1 - \frac{1}{m}}\Paren{\frac{\ln v_{(2),(1)}}{1 - \frac{1}{v_{(2),(1)}}} + t}}^{v_{(2),(1)}} \frac{1}{v\Paren{v_{(2),(1)} - v + t(m-1)}^2} \, \mathrm{d}v \\
    & = \frac{m(m-1)}{1 - \frac{1}{v_{(2),(1)}}} \SqBr{\frac{1}{(v_{(2),(1)} + t(m-1))(v_{(2),(1)} - v + t(m-1))}}_{\frac{v_{(2),(1)}}{m} + \Paren{1 - \frac{1}{m}}\Paren{\frac{\ln v_{(2),(1)}}{1 - \frac{1}{v_{(2),(1)}}} + t}}^{v_{(2),(1)}} \\
    & \hspace{13 pt} - \frac{m(m-1)}{1 - \frac{1}{v_{(2),(1)}}} \SqBr{\frac{\ln\Paren{\frac{v_{(2),(1)} + t(m-1)}{v} - 1}}{(v_{(2),(1)} + t(m-1))^2}}_{\frac{v_{(2),(1)}}{m} + \Paren{1 - \frac{1}{m}}\Paren{\frac{\ln v_{(2),(1)}}{1 - \frac{1}{v_{(2),(1)}}} + t}}^{v_{(2),(1)}} \\
    & \leq \frac{m}{1 - \frac{1}{v_{(2),(1)}}} \Paren{\frac{1}{t(v_{(2),(1)} + t(m-1))} + \frac{(m-1)\ln\Paren{\frac{v_{(2),(1)}}{t(m-1)}}}{(v_{(2),(1)} + t(m-1))^2} + \frac{(m-1) \ln (m-1)}{(v_{(2),(1)} + t(m-1))^2}} \\
    & = \frac{m}{1 - \frac{1}{v_{(2),(1)}}} \Paren{\frac{1}{t(v_{(2),(1)} + t(m-1))} + \frac{(m-1)\ln \Paren{\frac{v_{(2),(1)}}{t}}}{(v_{(2),(1)} + t(m-1))^2}} \\
    & \leq \frac{m}{1 - \frac{1}{v_{(2),(1)}}} \Paren{ \frac{1}{t(v_{(2),(1)} + t(m-1))} + \frac{(m-1)\ln v_{(2),(1)}}{(v_{(2),(1)} + t(m-1))^2}} \tag{for $t \geq 1$}
\end{align*}
We will use this bound for $t \geq 1$.
For $t \leq 1$, note that
\begin{align*}
    \PP\cSqBr{\sum_{j = 1}^m v_{(i),(j)} \geq v_{(2),(1)} + \frac{(m-1) \ln v_{(2),(1)}}{1 - \frac{1}{v_{(2),(1)}}} + t(m-1)}{v_{(2),(1)}}
        &\leq \PP\cSqBr{\sum_{j = 1}^m v_{(i),(j)} \geq v_{(2),(1)}}{v_{(2),(1)}} \\
        &\leq \frac{m \ln v_{(2),(1)}}{v_{(2),(1)} \Paren{1 - \frac{1}{v_{(2),(1)}}}} \tag{Markov; Lemma \ref{lemma:conditional-expectation-non-favorite}}
\end{align*}
The last inequality follows from the fact that we can treat $v_{(i),(1)}, \dots, v_{(i),(m)}$ as i.i.d. draws from $\ER$ conditioned on being at most $v_{(2),(1)}$.
Thus,
\begin{align*}
    &\EE\cSqBr{\Paren{\sum_{j = 1}^m v_{(i),(j)} - v_{(2),(1)} - \frac{(m-1) \ln v_{(2),(1)}}{1 - \frac{1}{v_{(2),(1)}}}}^+}{v_{(2),(1)}} \\
    & = (m-1) \int_0^\infty \PP\cSqBr{\sum_{j = 1}^m v_{(i),(j)} \geq v_{(2),(1)} + \frac{(m-1) \ln v_{(2),(1)}}{1 - \frac{1}{v_{(2),(1)}}} + t(m-1)}{v_{(2),(1)}} \,\mathrm{d}t \\
    &\leq \int_0^1 \frac{m(m-1) \ln v_{(2),(1)}}{v_{(2),(1)}\Paren{1 - \frac{1}{v_{(2),(1)}}}} \,\mathrm{d}t + \frac{m(m-1)}{1 - \frac{1}{v_{(2),(1)}}} \int_1^{v_{(2),(1)} - \frac{\ln v_{(2),(1)}}{1 - \frac{1}{v_{(2),(1)}}}} \Paren{ \frac{1}{t(v_{(2),(1)} + t(m-1))} + \frac{(m-1)\ln v_{(2),(1)}}{(v_{(2),(1)} + t(m-1))^2}} \, \mathrm{d}t \tag{$\sum_{j=1}^m v_{(i),(j)} \leq mv_{(2),(1)}$, so the integrand is 0 if $t \geq v_{(2),(1)} + \frac{\ln v_{(2),(1)}}{1 - 1/v_{(2),(1)}}$} \\
    &= \frac{m (m-1) \ln v_{(2),(1)}}{v_{(2),(1)}\Paren{1 - \frac{1}{v_{(2),(1)}}}} - \frac{m(m-1)}{1 - \frac{1}{v_{(2),(1)}}}\SqBr{\frac{\ln t}{v_{(2),(1)} + t(m-1)} + \frac{\ln v_{(2),(1)}}{v_{(2),(1)} + t(m-1)}}_1^{v_{(2),(1)} - \frac{\ln v_{(2),(1)}}{1 - \frac{1}{v_{(2),(1)}}}} \\
    &\leq \frac{2m(m-1) \ln v_{(2),(1)}}{v_{(2),(1)}\Paren{1 - \frac{1}{v_{(2),(1)}}}}
\end{align*}
We integrate the conditional expectation using the pdf of the second highest of $nm$ draws, which stochastically dominates $v_{(2),(1)}$.
The cdf of the second highest of $nm$ draws is
\[
    \Paren{1 - \frac{1}{v}}^{nm} + \frac{nm}{v}\Paren{1 - \frac{1}{v}}^{nm-1}
\]
The pdf is
\[
    \frac{nm(nm-1)}{v^3}\Paren{1-\frac{1}{v}}^{nm-2}
\]
Thus,
\begin{align*}
    \EE\SqBr{\Paren{\sum_{j = 1}^m v_{(i),(j)} - v_{(2),(1)} - \frac{(m-1) \ln v_{(2),(1)}}{1 - \frac{1}{v_{(2),(1)}}}}^+} 
        &\leq \int_1^\infty \frac{2m(m-1) \ln v}{v\Paren{1 - \frac{1}{v}}} \frac{nm(nm-1)}{v^3}\Paren{1-\frac{1}{v}}^{nm-2} \, \mathrm{d}v \\
        &= \frac{2m(m-1)(2 H_{nm} - 3)}{nm-2} \tag{Wolfram}
\end{align*}
so 
\[
    \sum_{i=3}^m \EE\SqBr{\Paren{\sum_{j = 1}^m v_{(i),(j)} - v_{(2),(1)} - \frac{(m-1) \ln v_{(2),(1)}}{1 - \frac{1}{v_{(2),(1)}}}}^+} \leq \frac{2m(m-1)(n-2)(2 H_{nm} - 3)}{nm-2} = O(m \ln (nm))
\]
It remains to bound $\EE\SqBr{\Paren{\sum_{j=2}^m v_{(2),(j)} - \frac{(m-1) \ln v_{(2),(1)}}{1 - \frac{1}{v_{(2),(1)}}}}^+}$, but note that by Lemmas \ref{lemma:expectation-second-highest-favorite} and \ref{lemma:conditional-expectation-non-favorite} and Jensen's inequality,
\begin{align*}
    \EE\SqBr{\Paren{\sum_{j=2}^m v_{(2),(j)} - \frac{(m-1) \ln v_{(2),(1)}}{1 - \frac{1}{v_{(2),(1)}}}}^+}
        &\leq \EE\SqBr{\sum_{j=2}^m v_{(2),(j)}} \leq \frac{(m-1) \ln (mn)}{1 - \frac{1}{mn}}
\end{align*}
\end{proof}

\subsection{Proofs from Section~\ref{sec:kfbic}} \label{sec:proofs_kfbic}

\kfatruthful*
\begin{proof}
    For each item $j$, once each bidder in $S$ is determined to be in or not in $S_j$, they simply face a posted price of either $H$ or $L$. Therefore, there is no incentive to misreport their values. If the bidder cannot control whether or not they are in $S_j$ (e.g. because the auctioneer knows this already), then the auction is truthful. Once the bidder cannot lie about whether or not they are in $S_j$ (e.g. their favorite item), then they cannot benefit from any further manipulation.
\end{proof}

\end{document}